\documentclass[journal, lettersize]{IEEEtran}
\usepackage{amsmath,amsfonts}
\usepackage{amsthm}
\usepackage{amssymb}
\usepackage{array}
\usepackage{arydshln}
\usepackage{subcaption}
\usepackage{textcomp}
\usepackage{stfloats}
\usepackage{url}
\usepackage{verbatim}
\usepackage{graphicx}
\usepackage{cite}
\usepackage{upgreek}
\usepackage{multirow}
\usepackage{multicol}
\usepackage{numprint}
\usepackage{dsfont}
\usepackage{nomencl}
\usepackage{tikz}
\usepackage{xfrac}
\usepackage[shortlabels]{enumitem}

\newtheorem{theorem}{Theorem}
\newtheorem{lemma}{Lemma}
\newtheorem{proposition}{Proposition}
\newtheorem{corollary}{Corollary}

\usepackage{listings,xfp}
\usepackage{anyfontsize}

\npdecimalsign{\ensuremath{.}}

\makeatletter
{\small 
\xdef\f@size@small{\f@size}
\xdef\f@baselineskip@small{\f@baselineskip}
\normalsize 
\xdef\f@size@normalsize{\f@size}
\xdef\f@baselineskip@normalsize{\f@baselineskip}
}
\newcommand{\smalltonormalsize}{%
  \fontsize
    {\fpeval{(\f@size@small+\f@size@normalsize)/2}}
    {\fpeval{(\f@baselineskip@small+\f@baselineskip@normalsize)/2}}%
  \selectfont
}
\makeatother

\makeatletter
{\footnotesize 
\xdef\f@size@footnotesize{\f@size}
\xdef\f@baselineskip@footnotesize{\f@baselineskip}
\small 
\xdef\f@size@small{\f@size}
\xdef\f@baselineskip@small{\f@baselineskip}
}
\newcommand{\footnotesizetosmall}{%
  \fontsize
    {\fpeval{(\f@size@footnotesize+\f@size@small)/2}}
    {\fpeval{(\f@baselineskip@footnotesize+\f@baselineskip@small)/2}}%
  \selectfont
}
\makeatother

\makeatletter
{\scriptsize 
\xdef\f@size@scriptsize{\f@size}
\xdef\f@baselineskip@scriptsize{\f@baselineskip}
\footnotesize 
\xdef\f@size@footnotesize{\f@size}
\xdef\f@baselineskip@footnotesize{\f@baselineskip}
}
\newcommand{\scriptsizetofootnotesize}{%
  \fontsize
    {\fpeval{(\f@size@scriptsize+\f@size@footnotesize)/2}}
    {\fpeval{(\f@baselineskip@scriptsize+\f@baselineskip@footnotesize)/2}}%
  \selectfont
}
\makeatother

\makenomenclature

\begin{document}

\title{On the Uplink and Downlink EMF Exposure and Coverage in Dense Cellular Networks:  \\ A Stochastic Geometry Approach}
\author{Quentin~Gontier,~\IEEEmembership{Graduate Student Member,~IEEE,} Charles~Wiame,~\IEEEmembership{Member,~IEEE,}  
Joe~Wiart,~\IEEEmembership{Senior Member,~IEEE,} François~Horlin,~\IEEEmembership{Member,~IEEE,} Christo~Tsigros, Claude~Oestges,~\IEEEmembership{Fellow,~IEEE,} and  Philippe~De~Doncker,~\IEEEmembership{Member,~IEEE}
\thanks{This work was supported by Innoviris under the Stochastic Geometry Modeling of Public Exposure to EMF (STOEMP-EMF)~grant.}
\thanks{Q. Gontier, F. Horlin and Ph. De Doncker are with Universit\'e Libre de Bruxelles, OPERA-WCG, Avenue Roosevelt 50 CP 165/81, 1050 Brussels, Belgium (quentin.gontier@ulb.be).} 
\thanks{C. Wiame is with NCRC Group, Massachusetts Institute of Technology, Cambridge, MA 02139 USA.} 
\thanks{J. Wiart is with Chaire C2M, LTCI, Télécom Paris, Institut Polytechnique de Paris, 91120 Palaiseau, France.} 
\thanks{C. Tsigros is with Department Technologies et Rayonnement, Brussels Environment, Belgium.} 
\thanks{C. Oestges is with ICTEAM Institute, Université Catholique de Louvain, Louvain-la-Neuve, Belgium.} 
}

\maketitle
\begin{abstract}
Existing studies analyzing electromagnetic field (EMFE) in wireless networks have primarily considered downlink communications. In the uplink, the EMFE caused by the user's smartphone is usually the only considered source of radiation, thereby ignoring contributions caused by other active neighboring devices. In addition, the network coverage and EMFE are typically analyzed independently for both the uplink and downlink, while a joint analysis would be necessary to fully understand the network performance and answer various questions related to optimal network deployment. This paper bridges these gaps by presenting an enhanced stochastic geometry framework that includes the above aspects. The proposed topology features base stations modeled via a homogeneous Poisson point process. The users active during a same time slot are distributed according to a mixture of a Matérn cluster process and a Gauss-Poisson process, featuring groups of users possibly carrying several equipments. In this paper, we derive the marginal and meta distributions of the downlink and uplink EMFE and we characterize the uplink to downlink EMFE ratio. Moreover, we derive joint probability metrics considering the uplink and downlink coverage and EMFE. These metrics are evaluated in four scenarios considering BS, cluster and/or intracluster densifications. Our numerical results highlight the existence of optimal node densities maximizing these joint probabilities.
\end{abstract}


\section{Introduction}
\IEEEPARstart{O}{ptimizing} and characterizing the performance of large-scale wireless communication networks is significantly challenging due to randomness of the propagation channel and irregularities in the infrastructure topology. Studying the uplink (UL) adds another layer of complexity because of the stochastic nature of user equipment (UE) locations, the limited power resources of battery-powered UEs and the UL power control strategies. Stochastic geometry (SG) has emerged as a valuable mathematical tool for assessing the system-level performance of large-scale networks by addressing randomness in the placement of base stations (BSs) and UEs \cite{Baccelli1997StochasticGA, tutorial}. 

\smallskip 

In this paper, SG is therefore employed in both UL and downlink (DL) to calculate the statistical distribution of two critical metrics, especially in the case of network densification: the signal-to-interference-plus-noise ratio (SINR) and the electromagnetic field exposure (EMFE). The former is directly connected to the data rate experienced by the users, while the latter is an important subject of discussion in the advent of 5G \cite{5Gpovchiaraviglio}. The massive densification of the cellular infrastructure and the use of new frequency bands associated to this new technology must indeed respect legal requirements in many countries in terms of EMF radiation. To make the study as complete as possible, we are considering the case where users are grouped into clusters, each with the option of having several pieces of equipment.

Regarding the sources of radiation, one can first mention the DL EMFE coming from base stations. A second source of radiation arises from active radio devices present in the vicinity of the user different from the smartphone of that user for which EMFE is being computed. Considered together, these two first sources constitute what we refer to as \textit{electromagnetic (EM) pollution}, which is addressed by specific legal constraints. A third source of EMFE stems from the smartphone carried by the user. Interactions between the EMF and the user's body then occur in the near-field, requiring a distinct approach. Typically, measurements in this context are conducted on phantoms or involve specific numerical methods \cite{chiaraviglio2023dominance}. For the scope of our discussion, we will therefore not further delve into this last source of radiation. This paper focuses instead on the EM pollution faced by active or idle users which is subject to specific legislation for network providers. In contrast, the user's phone is subject to legislation that phone builders must comply with.

In SG, homogeneous Poisson point processes (H-PPPs) are commonly employed to model BSs thanks to their tractability, combined to a realistic representation of the network. The H-PPP being motion-invariant (stationary and isotropic), performance metrics can be evaluated in all generality for a typical UE located at the origin. However, when it comes to the UL, a different approach is required since it involves the analysis of two point processes (PPs): one for the BSs and another one for the UEs. While H-PPPs can still be used for modeling BSs, PPPs are insufficient for accurately modeling UEs for two reasons. First, UEs can often be organized in clusters because of people gathering for various reasons and, each person possibly carrying several equipments communicating with the BS. Second, PPPs do not account for the correlations between locations of UEs and their serving BSs. Of all UEs in the network, only a fraction is relevant for the study of UL SINR. Indeed, because of the absence of intracell interference, only one UE can use a specific resource block (RB) at a time in each BS cell. If the active UEs are distributed according to a H-PPP, the resulting process in called Poisson-Voronoï (PV), but it is intractable. Consequently, to answer these two issues, we introduce a mixture of a Matérn cluster process (MCP) and a Gauss-Poisson process (GPP) to represent the UEs active in a given time slot, which are relevant for the study of UL EMF EMFE. The MCP consists of a H-PPP parent process, each of them having a Poisson-number of offspring uniformly distributed inside a disk of given radius centered at the parent location. Each point of the MCP can then become parent and has one or two offspring, the first at the parent's location, the second at a specific distance from the parent. The fraction of UEs utilizing the same RBs, and therefore relevant for UL SINR, are modeled using the approximation \cite{Haenggi17} which is numerically shown to be also valid if the UEs are originally distributed as a MCP+GPP mixture instead of a H-PPP.

Leveraging SG, the primary objective of this study is to address the following questions:
\begin{enumerate}
    \item How large is the mean UL EMFE compared to mean DL EMFE? How does the ratio of these two quantities evolve with the densification of BS and/or UEs and in the presence of clusters?
    \item For what proportion of users is EMFE above a specific threshold for what part of the time?
    \item What is the likelihood of having at the same time the UL EMFE below a given limit and the UL SINR above a specific threshold? Knowing this, for a given demand in terms of the number of UEs, what BS density choosing and what type of BS deploy? Additionally, once the network is deployed, what range of UE densities can this network handle while ensuring good performance?
    \item For users with an UL SINR requirement met and experiencing EMFE below a limit, for what DL SINR threshold can users expect a 80\% probability to meet DL coverage requirements as well for diverse BS or UE densities?
\end{enumerate}

These questions hold great significance to network providers and government authorities who are faced to the interplay between network optimization and EMFE regulation. They will guide the reader through the paper. 

\subsection{Related Works}
This section summarizes existing works connected to this study. Table~\ref{tab:comparison} compares all references mentioned within the present paper. The associated state of the art is divided into two parts in the following paragraphs: the modeling of UL communications in SG analyses and the modeling of EMFE in both DL and UL.

\subsubsection{Stochastic geometry for the SINR assessment in the UL}
Originally, SG was predominantly used to assess performance in the DL, including metrics like SINR coverage, data rate, and outage probability. As explained, in the UL, two PPs are in the game and several solutions exist to extract performance metrics. Some authors have used independent H-PPPs for both BS and UE locations \cite{Herath15, Nikbakht20, Elshaer_2016, Herath18, uplink_first, Direnzouplink16, Martin-Vega16, ElSawy14}. To address the missing correlation in the model, \cite{Singh15} suggested approximating the PP of UEs with an equivalent inhomogeneous Poisson point process (I-PPP) featuring a decreasing exponential intensity function when viewed from the typical BS. Nevertheless, this approach introduces errors in performance metrics. In \cite{Haenggi17}, a semi-empirical intensity function based on the real distance distribution in the network, observed from either the typical UE or the typical BS, was proposed to mitigate these errors. This method is employed to evaluate the meta distribution of UL SINR in \cite{uplink_meta}.

It is worth noting that some studies, such as \cite{Herath15, Nikbakht20, Elshaer_2016, Herath18, Singh15, uplink_meta, Bai16}, do not impose constraints on the maximum UL transmit power. This lack of constraint results in an overall increase in both signal and interference, while the noise remains at a constant level. Noise-limited analyses then take away from reality, such as in \cite{Herath15, Nikbakht20, uplink_meta, Bai16}. The associated errors are examined in \cite{Elshaer_2016}. 

\subsubsection{Stochastic geometry for incident power density assessment}
Application of SG has expanded, now including the assessment of incident power density (IPD) in wireless power transfer systems \cite{HE3}. From IPD, both the amount of power that can be harvested wireless power transfer systems in the context of the Internet-of-Things \cite{HE3}, and EMFE, can be derived. The difference between the two arises from the fact that harvested power should be maximized while EMFE should be minimized. Some studies calculate a joint complementary cumulative distribution function (CCDF) to achieve a trade-off between coverage and harvested power, as seen in \cite{TVC_SWIPT} and \cite{SWIPT_MIMO}. Recently, researchers have started investigating the evaluation of IPD for public health concerns related to EMFE, as EMF-aware systems aim to ensure that IPD remains low enough to respect legal EMFE limits. The first attempt to model EMFE in SG can be found in \cite{app10238753}, where an empirical propagation model is used for a 5G mMIMO network in the millimeter wave (mmWave) band. In \cite{GontierAccess}, the theoretical distribution of EMFE is compared to an experimental distribution obtained from measurements in an urban environment. EMFE is analyzed while considering max-min fairness DL power control in a 5G mMIMO network in \cite{power_control}. Additionally, \cite{9511258} explores the impact of the coexistence of sub-6GHz and mmWave BSs in networks on the EMFE. The EMFE has also been numerically assessed in an indoor environment in \cite{Shikhantsov20}, by using a methodology akin to the SG framework. This involves obtaining a large number of ray-launching realizations by employing a randomized arrangement of scatterers for each realization. The joint study of SINR coverage and EMFE is introduced in \cite{GontierTWC} for $\beta$-Ginibre point processes and I-PPPs, followed by works such as \cite{manhattan} for Manhattan networks and \cite{cell-free} for user-centric cell-free mMIMO networks. As for the meta distribution of EMFE, it has only been analyzed, and only for the DL, in our previous work \cite{gontier2024meta}, which examined the EMFE experienced by a passive user in a cellular network utilizing dynamic beamforming. These works focus on the DL. 

Regarding the UL, an evaluation of the EMFE caused by the user's smartphone using SG is presented in \cite{chen2023joint} where the authors independently study the impact of the network parameters on the EMFE and the SNR under the assumption that the network is noise-limited. A more comprehensive investigation \cite{Qin24} extends the analysis to include EMFE from BSs and active users, while also considering the impact of EMFE limits on network coverage. In \cite{Pardo24}, resource allocation optimization to maximize the number of connected users to a BS is explored via MC simulations, taking into account coverage requirements and EMFE considerations. Values of the output power levels of 4G UEs have been collected in \cite{Joshi17}, leading to the conclusion that the mean transmit power of a UE in an LTE network is less than 1\% of the maximal transmit power (23~dBm) in an urban environment. Measurements of the EMFE from the user's smartphone have also been collected for a 5G network in \cite{chiaraviglio2023dominance}. 

To the best of the authors' knowledge, no study has introduced a mixture of two PPs to represent nodes in a network. Only \cite{Qin24} has evaluated the impact of UL EMFE caused by other UEs, by deriving its marginal distribution but not its meta distribution or a joint EMFE/SINR distribution.

{\small
\begin{table*}[t]
\begin{center}
\begin{tabular}{|c||c|c|c|c||c|c|c|c|}
    \hline
    \multirow{3}{*}{Ref.} & \multicolumn{2}{c|}{Model} & \multicolumn{2}{c||}{Topology} & \multicolumn{4}{c|}{Calculated performance metrics of interest} \\
    \cline{2-9}
    & \multirow{2}{*}{UL/DL}  & Max UL & PV & \multirow{2}{*}{MCP+GPP} & IPD & IPD & Joint DL  & Joint UL \\
    &  & power & approx. & mixture & marg. & meta & IPD-SINR & IPD-SINR \\
    \hline
    \cite{Herath15, Nikbakht20, Herath18} & UL &  &  &  &  &  &  & \\
    \hline
    \cite{Elshaer_2016} & UL+DL &  &  &  &  &  &  & \\
    \hline
    \cite{uplink_first, Direnzouplink16, Martin-Vega16, ElSawy14} & UL & \checkmark &  &  &  &  &  & \\
    \hline
    \cite{Singh15} & UL+DL &  & \checkmark &  &  &  &  & \\
    \hline
    \cite{uplink_meta} & UL &  & \checkmark$^*$ &  &  &  &  & \\
    \hline
    \cite{Bai16} & UL &  & \checkmark &  &  &  &  & \\
    \hline
    \cite{HE3, app10238753, GontierAccess} & DL &  &  &  & \checkmark &  &  &  \\
    \hline
    \cite{TVC_SWIPT, SWIPT_MIMO} & DL &  &  &  &  &  & \checkmark &  \\
    \hline
    \cite{power_control} & DL &  &  &  & \checkmark &  &  &  \\
    \hline
    \cite{9511258, GontierTWC, manhattan, cell-free} & DL &  &  &  & \checkmark &  &  &  \\
    \hline
    \cite{gontier2024meta} & DL &  &  & \checkmark & \checkmark &  &  &  \\
    \hline
    \cite{Chen23} & UL+DL & \checkmark &  &  &  &  &  &  \\
    \hline
    \cite{Qin24} & UL+DL & \checkmark & \checkmark &  & \checkmark &  &  &  \\
    \hline
    This work & UL+DL & \checkmark$^*$ & \checkmark & \checkmark &\checkmark  & \checkmark & \checkmark & \checkmark \\
    \hline
\end{tabular}
\end{center}
\caption{Comparison between the relevant literature of SG for wireless telecommunication networks and this work. $^{*}$: The semi-empirical model of \cite{Haenggi17} is used. "marg." stands for marginal CDF.}
\label{tab:comparison}
\end{table*}
}

\subsection{Contributions}
Motivated by these considerations, this paper aims to answering to the four general questions of the introduction by (i) modeling a network of Poisson-distributed BSs and UEs with UL power control distributed as a mixture of a MCP and a GPP, utilizing the approximation of the PV model from \cite{Haenggi17} to approximate interfering UEs in the UL; (ii) assess the relative magnitude of UL EMFE in comparison DL EMFE and evaluate performance metrics, including SINR coverage and EMFE in both UL and DL; and (iii) investigate network performance following a network densification, determining the range of BS densities required to maintain an efficient network in terms of coverage and EMFE. The specific metrics introduced and analyzed are as follows:
\begin{enumerate}
    \item \textit{UL EM pollution}: Based on existing literature models, we derive the following mathematical expressions for UL EMFE, which are then compared with DL EMFE:
    \begin{itemize}
        \item[i.] Mean of the UL EMFE;
        \item[ii.] Marginal cumulative distribution function (CDF) of UL EMFE and of (global UL and DL) EM pollution;
        \item[iii.] Meta distribution of the UL EMFE and of the (global UL and DL) EM pollution.
    \end{itemize}
    \item \textit{Joint performance}: The second contribution consists of the analysis of the network performance from the points of view of EM pollution, UL coverage and DL coverage simultaneously. In particular, the following metrics are derived:
    \begin{itemize}
        \item[iv.] Joint CDF of UEC (\underline{U}L \underline{E}MFE and UL \underline{C}overage);
        \item[v.] Joint and conditional CDF of EMP-UDC (\underline{EM} \underline{P}ollution, \underline{U}L coverage and \underline{D}L \underline{C}overage).
    \end{itemize}
    \item \textit{Numerical results}: Capitalizing on the above metrics, the network performance is studied for two densification scenarios: 
    \begin{itemize}
        \item[(a)] Densification of both BSs and UEs with a constant UE/BS density ratio
        \item[(b)] Densification of UEs only
    \end{itemize}
    Optimal node densities are obtained as a function of the required coverage and EMFE limits. 
\end{enumerate}

\subsection{Structure of the Paper}
The paper is organized as follows: Section~\ref{sec:model} introduces the network topology and the system model. Section~\ref{sec:analytical_results} provides mathematical expressions for the performance metrics using SG. Numerical validation and analysis of the performance metrics in case of densification are provided in Section\ref{sec:results}. Finally, conclusions are given in Section~\ref{sec:conclusion}.

\section{System Model}
\label{sec:model}
\subsection{Topology}
\label{ssec:topology}
Let $\mathcal{B} \in \mathds{R}^2$ be the two-dimensional area where the network under study is located. The BSs locations are modeled as $\Psi_{b} = \Psi \cup \{X_0\}$ where $\Psi = \{X_i\} \in \mathcal{B}$ is a H-PPP with density $\lambda^{b}$ and $X_0$ is the BS located at the origin, as represented by the red crosses in Fig.~\ref{fig:Uplink_network_scheme}. The Slivnyak theorem states that $X_0$ becomes the typical BS under expectation over $\Psi_{b}$. The BSs are assumed to have the same technology, belong to the same network provider and transmit using the same sub-6~GHz frequency band $f^{d}$ with a bandwidth $B^{d}$.

\begin{figure}
    \centering
    \includegraphics[width=0.9\linewidth, trim={5cm, 11.5cm, 4cm, 12cm}, clip]{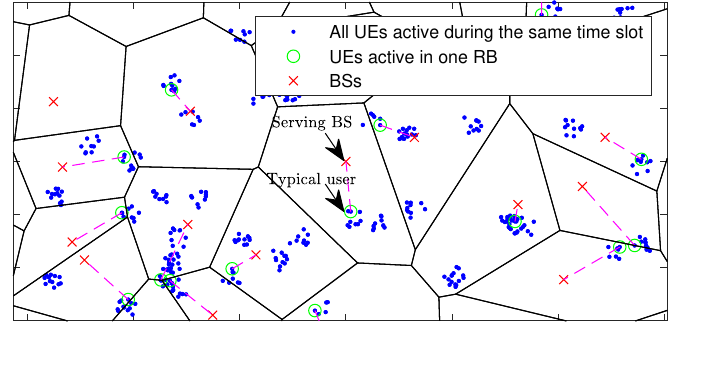}
    \caption{PV, model of type II \cite{Haenggi17}. The pink dashed lines correspond to the BS-UE association in one given RB. The black solid lines correspond to the edge of the Voronoï cells. The typical BS is located at the origin. 
    }
    \label{fig:Uplink_network_scheme}
\end{figure}
The entire UEs population scheduled during the same time slot is modeled as a mixture $\Phi_{u}$ of a MCP and a GPP, represented by the blue dots in Fig.~\ref{fig:Uplink_network_scheme}, and formerly defined in \cite{Azimi18} and \cite{Guo16}, respectively. The population of the grandparents is a H-PPP $\Phi^{[gp]}$ with density $\lambda^{[gp]}$. Each grandparent is the center of a cluster of a Poisson number of parents, uniformly distributed with density $\lambda^{[p]}$ within a radius $r_c$. At last, each parent has one or two offspring with probability $1-p^{[o]}$ and $p^{[o]}$, respectively. The first offspring is at the parent's location, the second, if it exists, is randomly distributed around the parent with some probability density function (PDF) $f^{[o]}$. To keep the UE's PP motion-invariant, $f^{[o]}$ only depends on the distance from the parent. To simplify calculations, $f^{[o]}$ is considered constant. The total density of UEs is then $\lambda^{u} = \lambda^{[gp]}\lambda^{[p]}(1+p^{[o]})\pi \left(r_c\right)^2$.

As this study is intended to be conservative, the typical user located at $Y_0 \in \Phi_u$ is assumed to be part of a cluster $C_0$ centered at $r^{[gp]}_0$. Its own equipments are not taken into account in the global EM pollution. The PDF of the distance $R_j^{[p]}$ between the typical user and a parent point in the cluster centered at $R^{[gp]}_i$, conditioned on $R^{[gp]}_i$, is given in \cite{qin23} by
{\small
\begin{multline}\label{eq:pdf_d_cluster}
    f_{R^{[p]}}(r_j^{[p]}|R_i^{[gp]}) = \\
    \begin{cases}
        2 r_j^{[p]}/r_c^2 &\!\text{if } 0 \leq r_j^{[p]} \leq r_c - R_i^{[gp]},\\
         \frac{2 r_j^{[p]}}{\pi r_c^2} \arccos\!\left(\!\frac{\left(r_j^{[p]}\right)^2+R_i^{[gp]}-r_c^2}{2r_j^{[p]}R_i^{[gp]}}\!\right)\! &\!\text{if } |r_c - R_i^{[gp]}| \leq r_j^{[p]} \leq r_c + R_i^{[gp]},\\
        0 &\!\text{otherwise.}
    \end{cases}
\end{multline}}

However, if all UEs in $\Phi_{u}$ contribute to UL EMFE, only a fraction of them uses the same subcarriers and is then relevant for studying UL SINR. Let's consider the specific subcarrier frequency $f^{u}$ with bandwidth $B^{u}$, not equal to $f^{d}$ in all generality, frequency division duplexing (FDD) being used for most technologies. Under the assumption of no intracell interference, as for the case of orthogonal frequency-division multiplexing (OFDM) signals, the only source of interference in the UL is the so-called intercell interference coming from other cells. The thinned PP of UEs using the same RBs is defined as $\Phi_{r} = \left\{Y_i \in \Phi_{u}: U(V(X_i) \cap \Phi_{u})\,\forall X_i \in \Psi_b\right\}$ \cite{Haenggi17} where $V(X)$ denotes the Voronoï cell of BS $X$, $U(V)$ with $V \in \mathcal{B}$ denotes one point chosen uniformly and randomly from $V$ and independently across different $V$. The resulting PP has a complex PV distribution. Cells with no active UE are called Crofton cells. The approximate density of $\Phi_{r}$ is $\lambda^{r} = \lambda^{b} (1-\nu)$ with $\nu = (\gamma/(\gamma+\delta))^\gamma$, $\delta = \lambda^{u}/\lambda^{b}$ and $\gamma = 7/2~$\cite{Haenggi17}. The set of interfering UEs is $\Phi_I = \Phi_{r} \backslash \{Y_0\}$.  The interfering UEs are represented by a green circle in Fig.~\ref{fig:Uplink_network_scheme}. A closest BS association policy is assumed. The link between UEs their serving BS is represented by the dashed pink line. 

Let us focus on the PPs $\Psi_b$ and $\Phi_{r}$. The typical UE located at $Y_0$ is served by the typical BS at $X_0$ and the distance between the two is $R_0 = ||Y_0||$. The distance between each UE $Y_i$ and its serving BS $X_i$ is denoted $R_i$ and the distance between the typical BS $X_0$ and an interfering UE $Y_i$ is denoted $D_i$. The distance between the typical UE and each BS $X_i$ is denoted $\rho_i$ and the distance between the typical UE and each UE $Y_i$ is denoted $\Tilde{D}_i$.  These notations can be seen in Fig.~\ref{fig:Sketch}.
\begin{figure}[!ht]
    \begin{center}
    \includegraphics[width=0.25\textwidth, trim={6cm, 22cm, 7cm, 1.5cm}, clip]{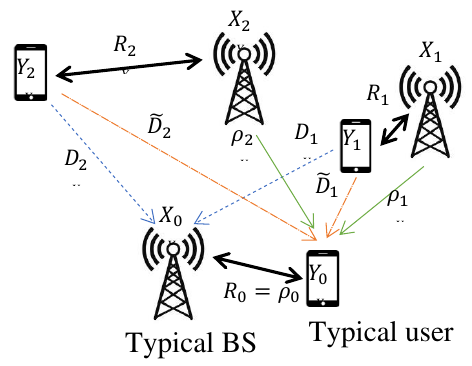}
    \end{center}
    \caption{Notations of the model. $X_i$ and $Y_i$ are respectively the locations of BS $i$ and UE $i$, with $X_0$ and $Y_0$ the typical BS-UE pair, $\rho_i$ is the distance between $X_i$ and $Y_0$, $D_i$ is the distance between $Y_i$ and $X_0$, $\Tilde{D}_i$ is the distance between $Y_i$ and $Y_0$ and $R_i$ is the distance between $X_i$ and $Y_i$.}
    \label{fig:Sketch}
\end{figure}

The spatial distribution of interfering UEs does not lead to tractable analytical expressions. In the case of a Poisson-distributed population of UEs, the model of type II \cite{Haenggi17} approximates the set of interfering UEs by an I-PPP with a density function depending on the distance to the typical BS, after fitting their pair-correlation function. The same technique applied to a population of UEs distributed with the mixture of a MCP and GPPs leads to the following density function:
For $r>0$,
\begin{multline}\label{eq:I_density_u}
    \lambda^{IU}(r) = \lambda^{r} \!\left(1-e^{-15.38 \sqrt{\lambda^{r}} r}+14.36 \lambda^{r} r^2 e^{-9.18 \lambda^{r} r^2}\right.\\
    \left.-13.95 \left(\lambda^{r}\right)^{3/1} r^3 e^{-5.02 \lambda^{r} r^2}\right).
\end{multline}
Similarly, the PDF of the distance between the typical BS and the typical UE can be approximated and is given by
\begin{equation}\label{eq:fR}
    f_{R_0}^\beta(r) = \beta\, 2\lambda^{b} \pi r e^{-\beta\, \lambda^{b} \pi r^2}, \ r > 0.
\end{equation}
with $\beta = 1.37$ for the mixture and $\beta = 1$ for a H-PPP. The associated CDF is
\begin{equation}\label{eq:FR}
    F_{R_0}^\beta(r) = 1- e^{-\beta\, \lambda^{b} \pi r^2}, \ r > 0.
\end{equation}

In the following, we will note $F_{R_0}^\beta(a, b) = F_{R_0}^\beta(b) - F_{R_0}^\beta(a)$. The link distances $R_y$ in the interfering cells are identically distributed as $R_0$, since all cells are statistically the same. However, they are not independent since the areas and shapes of neighboring cells are correlated, and $R_y$ cannot be larger than $D_y$. Hence we characterize the distribution of $R_y$ by conditioning on $D_y$. This results in the truncated Rayleigh distribution
\begin{equation}\label{eq:fRy}
    f_{R_y}(r|D_y) = \frac{f_{R_0}^\beta(r)}{f_{R_0}^\beta(D_y)}, \ 0 \leq r \leq D_y.
\end{equation}

All notations are listed in Table~\ref{tab:notations}.
{\small
\begin{table}[t]
\begin{center}
\begin{tabular}{|c|c|}
    \hline
    Notation & Meaning\\
    \hline
    $\Psi_{b}$, $\lambda^{b}$ & H-PPP of BSs, corresp. density\\
    $\Phi^{[gp]}$, $\lambda^{[gp]}$ & H-PPP of grandparent points, corresp. density\\
    $\Phi^{[p]}$, $\lambda^{[p]}$ & Set of parents inside a cluster, corresp. density\\
    $\Phi_{u}$, $\lambda^{u}$ & Population of active UEs, corresp. density \\
    \multirow{2}{*}{$\Phi_{r}$, $\lambda^{r}$} & Set of UEs using the same RBs,\\
    & corresp. approximate density  \\
    \multirow{2}{*}{$\Phi_{I}$, $\lambda^{IU}(r)$} & Set of UEs interfering in the UL,\\
    & corresp. density seen from $X_0$ \\
    $\rho_i$ & Distance between $Y_0$ and $X_i$ \\
    $\mathcal{B}(0,\tau)$ & 2D disk of radius $\tau$ \\
    $D_i$, $\Tilde{D}_i$ & Distance between $X_0$ and $Y_i$, between $Y_0$ and $Y_i$ \\
    $R_0$, $R_i$ & Distance between $X_0$ and $Y_0$, between $X_i$ and $Y_i$ \\
    $X_0$, $Y_0$ & Typical BS, typical UE \\
    $X_i \in \Psi_b$, $Y_i$ & BS labeled $i$, UE labeled $i$ \\
    \hdashline
    $\alpha$, $\kappa_{w}$ & PL exponent, $w$ PL intercept term \\
    $\upepsilon$ & FPC factor \\
    $\phi_{I_0^{d}}$ & CF of the DL interference seen from $X_0$\\
    $\phi_{\mathcal{P}^{w}}$ & CF of the $w$ EMFE\\
    $f^{w}$ & $w$ Carrier frequency with bandwidth $B^{w}$ \\
    $f^{[o]}$ & PDF of the distance between two offspring in a GPP\\
    $f_{R_0}^\beta\!(r)$, $F_{R_0}^\beta\!(r)$ & PDF of the distance $R_0$, corresp. CDF \\
    $f_{R_y}(r|D_y)$ & Conditional PDF of $R_i, i \neq 0$ \\
    $F_{R_0}^\beta(a,b)$ & $F_{R_0}^\beta(b)-F_{R_0}^\beta(a)$ \\
    $F^{w}_{\textrm{cov}}(T)$ & CCDF of the $w$ SINR \\
    $F^{w}_{\textrm{emfe}}(T_{e})$ & CDF of the $w$ EMFE \\
    $\mathcal{F}_{F_{\textrm{emfe}}}(T_{e}, s)$ & Meta distribution of EMFE\\
    $G_i^{b}$, $G_i^{u}$ & BS antenna gain, UE antenna gain \\
    $G^{u}(T_{c}^{u},T_e)$ & Joint CDF of the UEC \\
    $\!G(T_{c}^{u},T_{e}, T_{c}^d)\!$ & Joint CDF of EMP-UDC \\
    $|h_i^{w}|^2$ & $w$ channel fading gain \\
    $I_0^{w}$ & Aggregate $w$ interference \\
    $l_i^{d}$ & Channel power gain due to PL from $X_i$ to $Y_0$ \\
    $l_i^{u}$ & Channel power gain due to PL from $Y_i$ to $X_0$ \\
    $\Tilde{l}_i^{u}$ & Channel power gain due to PL from $Y_i$ to $Y_0$ \\
    $\mathcal{L}_{X}$ & Laplace transform of the random variable $X$ \\
    $N_{\sigma}$ & Thermal noise power \\
    $P_i^{w}$, $P_{i,r}^{w}$ & $w$ transmit EIRP, $w$ received power \\
    $P^{u}_{m}$ & Maximal UL transmit power \\
    $p^{[o]}$ & Probability that a parent has 2 offspring\\
    $r_e$ & Exclusion radius \\
    $r_c$ & Radius of each cluster of the MCP\\
    $S_0^{w}$ & Useful $w$ signal \\
    $z$ & Height difference between a BS and a UE \\
    $[f(x)]_{x=a}^{x=b}$ & $f(b)-f(a)$\\
    \hline
\end{tabular}
\end{center}
\caption{Notations used in the text. Notations above the dashed line are relative to network topology and distances. $w$ stands for $u$ (UL) or $d$ (DL).}
\label{tab:notations}
\end{table}
}

\subsection{Propagation Model}
\label{ssec:propagation_model}
The propagation model is defined as
\begin{equation}\label{eq:prop_UL}
    P_{i, r}^{u} = P_i^{u} G_i^{u} G_i^{b} |h_i^{u}|^2 l_i^{u} 
\end{equation}
for the UL and 
\begin{equation}\label{eq:prop_DL}
    P_{i, r}^{d} = P^{d} G_i^{b} G_i^{u} |h_i^{d}|^2 l_i^{d}
\end{equation}
for the DL, where $P_{i, r}^{u}$ (resp. $P_{i, r}^{d}$) is the UL (resp. DL) received power from UE $Y_i$ (resp. BS $X_i$), $P_i^{u}$ (resp. $P_{i}^{d}$) is the transmit power of UE $Y_i$ (resp. BS $X_i$), $G_i^{b}$ (resp. $G_i^{u}$) is the BS (resp. UE) antenna gain that is assumed omnidirectional, $|h_i^{u}|^2$ (resp. $|h_i^{u}|^2$) accounts for the UL (resp. DL) fading and
\begin{equation}
    l_i^{u} = l^{u}(D_i) = \kappa_u^{-1} \left(D_i^2+z^2\right)^{-\alpha/2},
\end{equation} 
\begin{equation}
    \Tilde{l}_i^{u} = \Tilde{l}^{u}(\Tilde{D}_i) = \kappa_u^{-1} \Tilde{D}_i^{-\alpha},
\end{equation} 
\begin{equation}
    l_i^{d} = l^{d}(\rho_i) = \kappa_d^{-1} \left(\rho_i^2+z^2\right)^{-\alpha/2},
\end{equation} 
are respectively the channel power gain due to path loss (PL) from UE $Y_i$ to typical BS $X_0$, from UE $Y_i$ to typical UE $Y_0$ and from BS $X_i$ to typical UE $Y_0$ with exponent $\alpha > 2$, $z > 0$ is the difference of height between $X_i$ and $Y_i$ (identical for all), $\kappa_d  \!= \! (4\pi f^{d}/c)^2$ and $\kappa_u  \!= \! (4\pi f^{u}/c)^2$ where $c$ is the speed of light. In the following, an abuse of notation will be used by writing $P_i^{x}$ instead of $P_i^{x} G_i^{x}$. We assume that $|h_i^{u}|^2$ and $|h_i^{d}|^2$ are uncorrelated due to fast time fluctuation of the channel. The same DL transmit power $P^{d}$ is assumed for all active BSs and the same antenna gains are assumed for UEs and BSs.

\subsection{Power control}
The 3rd Generation Partnership Project (3GPP) introduced fractional power control (FPC) for the Physical Uplink Shared Channel (PUSCH) to mitigate intercell UL interference and improve SINR coverage for cell-edge UEs \cite{R1-073224}. This FPC strategy, first proposed in \cite{Whitehead93} to reduce signal-to-noise ratio (SNR) variance and achieve target SNR, adjusts UE transmit power based on measurements of signal power by both the UE and the BS. The closed-loop FPC adjusts the fractional PL compensation (coefficient $\upepsilon \in [0,1]$) controlled by the network, resulting in a transmit power from UE $i$ given by
\begin{equation}
P_{i}^{u} = \min{\left(P^{u}_{0} l_i^{-\upepsilon}, P^{u}_{m}\right)}
\end{equation}
where $P^{u}_{0}$ represents the open-loop power whose range of values is calculated in \cite{R1-074850}. A maximum emitting power $P^{u}_{m}$ is considered in the model. Using FPC instead of full power control ($\upepsilon = 1$) allows higher PL users to operate at a lower SINR, reducing intercell interference. However, the optimal FPC coefficient must be found to balance a high SINR for cell-edge users against low intercell interference. Many studies have investigated the impact of FPC on network SINR, SNR, or signal-to-interference ratio (SIR), through simulations at the scale of one cell or at a larger scale modeling the network as an hexagonal grid \cite{Ubeda08, Simonsson08, Mullner09, Xiao06, Rao07}. They reveal that optimal values of $\upepsilon$ are between 0.4 and 0.6. To simplify the calculations, the emit power of UEs inside a given cluster are considered equal and the PL to compensate is the one of that a UE would have at the center of the cluster.

\subsection{Network Size}
The mathematical expressions derived in the following are defined for a circular area $\mathcal{B}$ of radius $\tau$ located in the $xy$ plane. The calculations take an exclusion radius $r_e \geq \min\left\{\kappa_u, \kappa_d\right\}^{-1/\alpha}$ into account, representing a non-publicly accessible area around the BSs and guaranteeing that the far-field conditions are met and to avoid singularities in the vicinity of the UE when calculating UL EMFE.

\subsection{Quantities of Interest}\label{ssec:quantities}

To simplify the notations, without loss of generality, a unitary UE gain $G^{u} \!=\! \numprint[dB]{0}$ will be assumed and $P_i^{u}$ and $P^{d}$ will be the UL and DL effective isotropic radiated power (EIRP), respectively. Let $S_0^{u} \! = \! P_0^{u}\,G^{b}\,|h_0^{u}|^2\,l_0^{u}$ be the useful UL received power, let $S_0^{d} \! = \! P^{d}\,|h_0^{d}|^2\,l_0^{d}$ be the useful DL received power, 
let $I_0^{u}  \!=  \!\sum_{i \in \Phi_I} P_i^{u}\,G^{b}\, |h_i^{u}|^2\, l_i^{u}$ be the aggregate UL interference and let $I_0^{d} \! = \! \sum_{i \in \Psi_I} P^{d}\, |h_i^{d}|^2\, l_i^{d}$ be the aggregate DL interference. Based on these definitions, the UL SINR at the typical BS, conditioned on the distance to its associated UE, is given by
\begin{align}
    \textrm{SINR}_{0}^{u} = \frac{S_0^{u}}{I_0^{u} + N_{\sigma}^{u}}
\end{align}
where $N_{\sigma}^{u}  \!= \! k\,B^{u}\,T\,\mathcal{F}^b$ is the thermal noise power with $k$ the Boltzmann constant, $T$ the standard temperature and $\mathcal{F}^b$ the BS noise figure. The DL SINR at the typical UE, conditioned on the distance to its associated UE, is given by
\begin{align}
    \textrm{SINR}_{0}^{d} = \frac{S_0^{d}}{I_0^{d} + N_{\sigma}^{d}}.
\end{align} 
where $N_{\sigma}^{d}  \!=  \!k\,B^{d}\,T\,\mathcal{F}^u$ with a UE noise figure $\mathcal{F}^u$. The UL EMFE experienced by the typical user and caused by all active UEs except the user's own smartphone is
\begin{align}\label{eq:exp_power_UL}
    \mathcal{P}^{u} = \sum_{i \in \Phi_{u} \backslash \{Y_0\}} P_i^{u}\,|\Tilde{h_i}^{u}|^2\, \Tilde{l}_i^{u}.
\end{align}
The fading coefficients $|h_i^u|^2$, $|h_i^d|^2$, $|\Tilde{h}_i^u|^2$ and $|\Tilde{h}_i^d|^2$ are all independent and identically distributed random variables. This independence stems from the application of frequency division duplexing and the poor correlation between channels at the utilized frequencies. The fading coefficients follow an exponential distribution with unit mean. For simplicity in the notation, they will be written $|h_i|^2$ without distinction in the following. To simplify the notations, we will write $\Bar{S}^{d}(r)  \!=  \!P^{d}\, l^{d}(r)$ and $\Bar{S}^{u}(r_1,r_2)  \!=  \!P^{u}(r_1)\,G^{b}\, l^{u}(r_2)$ in the following.
The DL EMFE experienced by the typical user and caused by the BSs of the network is
\begin{align}\label{eq:exp_power_DL}
    \mathcal{P}^{d} = \sum_{i \in \Psi_{b}} P^{d}\,|\Tilde{h_i}^{d}|^2\, l_i^d.
\end{align}

The power EMFE can be converted into a total IPD as
\begin{equation}\label{eq:IPD_def}
    \mathcal{S}^w = \kappa/(4\pi)\mathcal{P}^w 
\end{equation}
for $w = {u,d}$ and, finally, into a root-mean-square electric field strength in V/m as
\begin{equation}\label{ex:V/m}
    E[\textrm{V/m}] = \sqrt{120 \pi \mathcal{S}}.
\end{equation}

\section{Mathematical Results}
\label{sec:analytical_results}
\subsection{Analysis of EM Pollution}
\label{ssec:EMF_pollution}

The CDF of DL EMFE has already been calculated in \cite{GontierAccess} using Gil Pelaez' theorem and is therefore given as a final result in Lemma~\ref{lem:DL_EMF}.
\begin{lemma}\label{lem:DL_EMF}
    The CDF of the DL EMFE seen by the typical user is

\begin{align}
\begin{split}
    F_{\textrm{emfe}}^{d}(T_{e}) &= \mathbb{P}\left[\mathcal{P}^{d} < T_{e}\right] \\
    &= \frac{1}{2} - \int_0^\infty \!\textrm{\normalfont Im}\left[\phi_{\mathcal{P}^{d}}(q) e^{-j q T_{e}}\right]\frac{1}{\pi q} \, dq
\end{split}
\end{align}
where $\phi_{\mathcal{P}^{d}}(q)$ is the DL EMFE characteristic function (CF)
{\small
\begin{multline}
    \phi_{\mathcal{P}^{d}}(q) = \mathbb{E}\left[\exp\left(-j q \left(\mathcal{P}^{d}\right)\right)\right]\\
    = \exp\left(\pi \lambda^b \!\left[\left(x^2+z^2\right) _2F_1\left(1, \frac{2}{\alpha}, 1+\!\frac{2}{\alpha}, \frac{-j}{q \Bar{S}^{d}(x)}\right)\right]_{x = r_e}^{x = \tau}\right)\!.
\end{multline}}
\end{lemma}

The passive UL EMFE, caused by all users emitting simultaneously in the network, can be derived similarly using the Gil-Pelaez inversion theorem. It requires the knowledge of the CF of the passive UL EMFE. It is given in Theorem~\ref{eq:expucdf}.

\begin{theorem}\label{eq:expucdf}The CDF of UL EMFE seen by the typical user in a network where UEs are distributed according to a MCP+GPP mixture for the propagation model in \eqref{eq:prop_UL}, without considering its own equipment, is
\begin{align}
    F_{\textrm{emfe}}^{u}(T_{e}) &= \mathbb{P}\left[\mathcal{P}^{u}\! < T_{e}\right]= \frac{1}{2} - \! \int_0^\infty \!\textrm{\normalfont Im}\left[\phi_{\mathcal{P}^{u}}(q) e^{-j q T_{e}}\right]\frac{1}{\pi q} dq 
\end{align}
where $\phi_{\mathcal{P}^{u}}(q) = \phi_{\mathcal{P}^u}^{intra}(q) \phi_{\mathcal{P}^u}^{inter}(q)$,
{\small
\begin{equation}\label{eq:phi_P_inter}
    \phi_{\mathcal{P}^u}^{inter}(q) = \exp\left(2\pi\lambda^{[gp]} \int_{r_c}^{\tau} \left(G_{\textrm{MCP}}^{[x]}(v)-1\right)x\,dx\right),
\end{equation}}
{\small
\begin{equation}\label{eq:phi_P_intra}
    \phi_{\mathcal{P}^u}^{intra}(q) = \int_{-r_c}^{r_c} G_{\textrm{MCP}}^{[x]}(v) f_{R^{[p]}}(x|x < r_c) dx-e^{-\lambda^{[p]}\pi r_c^2},
\end{equation}}
{\smalltonormalsize
\begin{equation}
    G_{\textrm{MCP}}^{[x]}(v) = \!\int_{r_e}^{\tau}G_{\textrm{MCP}}^{[x]}(v|r)f_{R_0}^1(r)dr,
\end{equation}}
{\footnotesizetosmall
\begin{equation}
    G_{\textrm{MCP}}^{[x]}(v|R) \!= \!\exp\!\left(\!\lambda^{[p]}\pi r_c^2\!\left(\int_0^{r_c+x}\!G_{\textrm{GPP}}^{[u]}(v|R) f_{R^{[p]}}(u|x)du\!-\!1\right)\!\right),
    \end{equation}}
    $G_{\textrm{GPP}}^{[u]}(v|R) = (1-p^{[o]})v(q, u|R)+p^{[o]}v(q, u|R)\int_{-r_c}^{r_c}v(q, u+y|r)f^{[o]}(y)dy$ and $v(q, u|R) = (1-j q P^{u}(r) \Tilde{l}^u(u))^{-1}$.
\end{theorem}
\begin{proof}
    The proof of the CF of UL EMFE is provided in Appendix~\ref{sec:expucdfproof}.

\end{proof}
The previous mathematical expressions considered UL and DL EMFEs separately. Calculations can also be made for the EM pollution considering both UL and DL. The total EM pollution is given by the sum of UL and DL EMFEs. In our model, both sources of EMFE can be seen as totally independent sources.

\begin{theorem}\label{th:exptot}The CDF of total EM pollution experienced by the typical user in a network where BSs are distributed according to a H-PPP and UEs are distributed according to a MCP+GPP mixture for the propagation model in \eqref{eq:prop_UL} and in \eqref{eq:prop_DL}, without considering the EMFE of its own equipment is
\begin{multline}
\begin{split}
    &F_{\textrm{emfe}}^{tot}(T_{e}) = \mathbb{P}\left[\mathcal{P}^{d}+\mathcal{P}^{u} < T_{e}\right]\\
    &\quad= \frac{1}{2} - \int_0^\infty \textrm{\normalfont Im}\left[\phi_{\mathcal{P}^{d}}(q)\, \phi_{\mathcal{P}^{u}}(q) \, e^{-j q T_{e}}\right]\,\frac{1}{\pi q} \, dq 
\end{split}
\end{multline}
\end{theorem}
\begin{proof}
    Theorem~\ref{th:exptot} is derived using the Gil-Pelaez theorem. The PPs $\Psi_b$ and $\Phi_u$ are independent. Although $\mathcal{P}^{d}$ is independent of $\Phi_u$, $\mathcal{P}^{u}$ does exhibit dependence on $\Psi_b$ due to power control mechanisms. This correlation is inherently accounted for through the PDF \eqref{eq:fR} approaching the distance between each UE and its serving BS, as already done in Appendix~\ref{sec:expucdfproof}. By  incorporating this dependence within the PDF, $\mathcal{P}^{d}$ and $\mathcal{P}^{u}$ can be treated as independent in the analysis. From the definition of the CF, we have
    {\smalltonormalsize
    \begin{align}
    \begin{split}
        &\phi_{\mathcal{P}^{tot}}(q) = \mathds{E}\left[e^{j q \mathcal{P}^{tot}}\right] = \mathds{E}\left[\exp\left(j q \left(\mathcal{P}^{d}+\mathcal{P}^{u}\right)\right)\right]\\
        \quad&= \mathds{E}\left[\exp\left(j q \mathcal{P}^{d}\right)\exp\left(j q \mathcal{P}^{u}\right)\right] = \phi_{\mathcal{P}^{d}}(q)\, \phi_{\mathcal{P}^{u}}(q).
    \end{split}
    \end{align}}
\end{proof}
The mean IPD has already been calculated in \cite{9511258} and more specifically for the DL EMFE in \cite{GontierAccess} for a H-PPP, then for other PPs in \cite{GontierTWC}, using the Campbell theorem. An alternative method of calculation consists of taking the derivative of the CF, $\Bar{\mathcal{P}} = \frac{d \phi_{\mathcal{P}}(j t)}{dt}\Large|_{t=0}$. As the two alternatives do not present new challenges, the final result is given below as a corollary. 

\begin{corollary}\label{cor:mean_tot}
    The mean total EM pollution experienced by the typical user in a network where BSs are distributed according to a H-PPP and UEs are distributed according to a MCP+GPP mixture for the propagation model in \eqref{eq:prop_UL} and in \eqref{eq:prop_DL}, without considering the EMFE of its own equipment is $\Bar{\mathcal{P}^{tot}} = \Bar{\mathcal{P}^u} + \Bar{\mathcal{P}^d}$
    where
    \begin{equation}\label{eq:mean_exp_DL}
        \Bar{\mathcal{P}^d} = P^{d}\, 2\pi\lambda^{r} \frac{l^d(r_e)(r_e^2+z^2)-l^d(\tau)(\tau^2+z^2)}{\alpha-2};
    \end{equation}
    {\smalltonormalsize
    \begin{multline}\label{eq:mean_exp_UL}
        \Bar{\mathcal{P}^u} = 2\pi^2\lambda^{[gp]}\lambda^{[p]} r_c^2 \Bar{P^u} \int_{r_e}^\tau x \int_{0}^{r_c+x} \left(\Tilde{l}^u(u)\right.\\
        \left.+p^{[o]} \int_{-r_c}^{r_c} \Tilde{l}^u(u+y)f^{[o]}(y)dy\right) f_{R^{[p]}}(u|x) du dx,
    \end{multline} }
     $\Bar{P^u} = \int_{r_e}^\tau P^u(r)f_{R_0}^1(r) dr = \int_{r_e}^{r_{m}} P^{u}_{0} l^{-\upepsilon}(r)\,f_{R_0}^1(r)dr
        +P^{u}_{m} \, F_{R_0}^1(r_m, \tau) = P_{FPC}
        +P^{u}_{m} \, F_{R_0}^\beta(r_e, r_{m})$ with 
    \begin{multline}
        P_{FPC} = P_0^{u}\pi \lambda^{b} \beta e^{\beta z^2}\\
        \times \left[l^{-\upepsilon}(x)\left(x^2+z^2\right) E_{\frac{-\alpha\upepsilon}{2}}\left(\pi \lambda^{b} \beta\left(x^2+z^2\right)\right)\right]_{x=r_m}^{x=r_e}
    \end{multline}
    where $E_n(x)$ is the exponential integral function
    \begin{equation}
        E_n(x) = \int_1^\infty e^{-x t}/t^n dt
    \end{equation}which converges and $[f(x)]_{x=a}^{x=b} = f(b)-f(a)$.
\end{corollary}
It is worth noting that if the UEs are distributed according to a H-PPP, the mean UL EMFE simplifies to $\Bar{\mathcal{P}^u} = 2\pi\lambda^{u} \Bar{P^u}\frac{\Tilde{l}^{u}(r_e)r_e^2-\Tilde{l}^{u}(\tau)\tau^2}{\alpha-2}$. 

\subsection{Meta Distribution of EMFE}

The meta distribution of EMFE provides a comprehensive statistical characterization of the variability in EMFE levels experienced by users in a wireless network. Unlike the traditional CDF, the meta distribution delves into the distribution of individual users' EMFE, distinguishing the sources of EMFE variability and therefore offering insights into the tail behavior. In the model of this paper, the sources of variability are (i) the BS locations, (ii) the cluster locations, (iii) the location of UEs inside clusters and (iv) time fading in the propagation channel. In a real-world network, BS locations are fixed and cluster locations are less likely to fluctuate as fast as sources (iii) and (iv). Since fluctuations in the PL attenuation caused by sources (i) and (ii) only depend on the typical user's location, the meta distribution will enable to distinguish sources (i) and (ii) from sources (iii) and (iv). In other words, the CDF of EMF conditioned on one PP realization of BSs $\Psi_b$ and cluster centers $\Phi^{[gp]}$ is introduced, and denoted as $F_{\textrm{emfe}}(T_e|\Psi_b, \Phi^{[gp]}) = \mathbb P\left[\mathcal{P}^{tot} < T_e|\Psi_b, \Phi^{[gp]}\right]$. For $s \in [0, 1]$, the meta distribution of EMFE is then defined as
\begin{equation}\label{eq:meta_def}
\mathcal{F}_{F_{\textrm{emfe}}^{tot}}(T_{e}, s) = \mathbb P_{\Psi_b, \Phi^{[gp]}} \left[F_{\textrm{emfe}^{tot}}(T_e|\Psi_b, \Phi^{[gp]}) > s\right].
\end{equation}
Here, $\mathbb P_{\Psi_b, \Phi^{[gp]}}[\cdot]$ signifies that the probability operator is taken over all spatial realizations of the PP. When considering a large network, selecting a new realization of the PP around the centric PU at (0,0) is akin to observing the realization of the PP from the point of view of a distant PU in the network. Consequently, employing $\mathbb{P}_{\Psi_b, \Phi^{[gp]}}$ essentially evaluates performance for numerous users uniformly distributed across the network. In terms of interpretation, $\mathcal{F}_{F_\textrm{emf}^{tot}}(T_{e}, s) = x$\% means that $x$\% of PUs in the network experience an EMFE below the threshold $T_e$ for at least a fraction $s$ of the time. 

\begin{lemma}\label{lem:CDF_cond}
    The CDF of total EM pollution conditioned on the realizations $\Psi_b$ and $\Psi^{[gp]}$ for the propagation models in \eqref{eq:prop_DL} and \eqref{eq:prop_UL} can be written by
{\smalltonormalsize
\begin{multline}\label{eq:CDF_cond}
    F_{\textrm{emfe}^{tot}}\left(T_{e}|\Psi_b, \Phi^{[gp]}\right) \\
    = \frac{1}{2}-\int_{0}^{\infty} \!\textrm{\normalfont Im}\left[\phi_{\mathcal{P}^{tot}}(q|\Psi_b, \Phi^{[gp]})\,e^{-jqT_{e}}\right]\!\frac{1}{\pi q}\,dq
\end{multline}}
where $\phi_{\mathcal{P}^{tot}}(q|\Psi_b, \Phi^{[gp]})\! =\! \phi_{\mathcal{P}^{d}}(q|\Psi_b) \phi_{\mathcal{P}^{u}}(q|\Psi_b, \Phi^{[gp]})$,
{\smalltonormalsize
\begin{equation}\label{eq:cf_pd_cond}
    \phi_{\mathcal{P}^{d}}(q|\Psi_b) = \prod_{i \in \Psi_b}\left(1- j q P^d l_i^d\right)^{-1},
\end{equation}}
{\smalltonormalsize
\begin{multline}\label{eq:cf_pu_cond}
    \phi_{\mathcal{P}^{u}}(q|\Psi_b, \Phi^{[gp]}) =\! \left(G_{\textrm{MCP}}^{[R^{[gp]}_0]}(v) -e^{-\lambda^{[p]}\pi r_c^2}\right)\! \\
    \times\prod_{i \in \Phi^{[gp]} \setminus{R^{[gp]}_0}}\! G_{\textrm{MCP}}^{[X_i]}(v|R_i)
\end{multline}}
\end{lemma}
\begin{proof}
The expression in \eqref{eq:CDF_cond} is derived by applying Gil-Pelaez' theorem. The CFs are obtained following Appendix~\ref{sec:expucdfproof}. 
\end{proof}

The meta distribution of the EMFE \eqref{eq:meta_def} can be rewritten as $\mathbb P_\Psi \left[\ln(F_{\textrm{emfe}}^{tot}(T_e|\Psi_b, \Phi^{[gp]})) > \ln(s)\right]$. Using the Gil-Pelaez theorem, it follows that 
{\smalltonormalsize
\begin{equation}\label{eq:meta_GP}
    \mathcal{F}_{F_{\textrm{emfe}}^{tot}}(T_{e}, s) = \frac{1}{2}+\frac{1}{\pi} \int_0^{\infty} \textrm{\normalfont Im}\left[e^{-j q \ln(s)}\mathcal{M}_{jq}(T_e)\right] q^{-1}\,dq
\end{equation}}where $\mathcal{M}_{jq}(T_e) = \mathbb E_{\Psi_b, \Phi^{[gp]}}\left[\left(F_{\textrm{emfe}^{tot}}(T_e|\Psi_b, \Phi^{[gp]})\right)^{jq}\right]$ is the moment of order $jq$ of the conditional CDF of EMFE. The moments are not mathematically tractable for any value of~$q$. Although an exact derivation of the meta distribution of EMFE is not possible, a close approximation is provided by the beta distribution, leveraging the methodology of \cite{uplink_meta}, based on the knowledge of the first two moments of the conditional CDF, for which a tractable expression can be found.
\begin{proposition}\label{prop:beta_approx}The meta distribution of EMFE in \eqref{eq:meta_def} can be approximated by
{\smalltonormalsize
\begin{equation}\label{eq:approx_beta}
    \Tilde{\mathcal{F}}_{F_{\textrm{emfe}}^{tot}}(T_{e}, s) = 1-\mathcal{I}_s\left[\Tilde{\alpha}_{\mathcal{S}}(T_e), \Tilde{\beta}_{\mathcal{S}}(T_e)\right]
\end{equation}}
where $\mathcal{I}_s(a,b) = \frac{B(s;a,b)}{B(a,b)}$ is the regularized incomplete beta function, $B(s;a,b) = \int_0^s t^{a-1}(1-t)^{b-1}$ is the incomplete beta function, $B(a,b) = B(1;a,b)$ the beta function and $\Tilde{\alpha}_{\mathcal{S}}(T_e)$ and $\Tilde{\beta}_{\mathcal{S}}(T_e)$ are the parameters of the beta distribution obtained via moment matching. Under the condition $\mathcal{M}_{1}(T_e)< \mathcal{M}_{2}(T_e)$, their expressions are given by
{\smalltonormalsize
\begin{equation}
    \Tilde{\alpha}_{\mathcal{S}}(T_e) = \mathcal{M}_{1}(T_e) \left(\frac{\mathcal{M}_{1}(T_e)\left(1-\mathcal{M}_{1}(T_e)\right)}{\mathcal{M}_{2}(T_e)-\mathcal{M}_{2}^2(T_e)}-1\right);
\end{equation}}
{\smalltonormalsize
\begin{equation}
    \Tilde{\beta}_{\mathcal{S}}(T_e) = \left(1-\mathcal{M}_{1}(T_e)\right) \left(\frac{\mathcal{M}_{1}(T_e)\left(1-\mathcal{M}_{1}(T_e)\right)}{\mathcal{M}_{2}(T_e)-\mathcal{M}_{2}^2(T_e)}-1\right).
\end{equation}}
\end{proposition}

The first-order moment is $\mathcal{M}_{1}(T_e)\! = \!F_{\textrm{emfe}}^{tot}(T_e)$ given in Theorem~\ref{th:exptot}, which is directly obtained from $\phi_{\mathcal{P}^{d}}(q) \!= \!\mathbb{E}_{\Psi_b}\left[\phi_{\mathcal{P}^{d}}(q|\Psi_b)\right]$ and $\phi_{\mathcal{P}^{u}}(q)\! =\! \mathbb{E}_{\Psi_b, \Phi^{[gp]}}\left[\phi_{\mathcal{P}^{u}}(q|\Psi_b, \Phi^{[gp]})\right]$. The second-order moment is given in Theorem~\ref{th:M2}.

\begin{theorem}\label{th:M2} The second-order moment of the conditional distribution of EMFE \eqref{eq:CDF_cond} is given by
{\smalltonormalsize
\begin{equation}\label{eq:M2}
    \mathcal{M}_{2}(T_{e}) = -{1}/{4} + F_\textrm{emfe}^{tot}(T_{e})
    + {\pi^{-2}}\, {\Omega(T_e)}
\end{equation}}
where 
{\smalltonormalsize
\begin{equation}\label{eq:Omega}
    \Omega(T_e) = \int_{0}^{\infty}\int_{0}^{\infty}{\omega(T_e, q, q')}\,q^{-1}\,q'^{-1}\,dq\,dq',
\end{equation}}
{\small
\begin{equation}
        \omega(q, q';T_{e}) = \frac{1}{2}\,\textrm{\normalfont Re}\!\left[\gamma_-(q, q')\,e^{-jT_{e}(q-q')}\!+\!\gamma_+(q, q')\,e^{-jT_{e}(q+q')}\right],
\end{equation}}$\gamma_\pm(q, q') = \gamma_+^d(q, \pm q') \gamma_+^u(q, \pm q')$\!,
{\small
\begin{multline}
    \gamma_+^d(q, q') \\
        = \exp\!\left(2\pi \lambda^b\int_{r_e}^{\tau}\left(\frac{1}{1- j q P^d l(x)^d}\frac{1}{1- j q' P^d l(x)^d}-1\right)dx\right),
\end{multline}}
{\small
\begin{multline}\label{eq:gamma+u}
    \gamma_+^u(q, q') = e^{-\lambda^{[p]}\pi r_c^2} \phi_{\mathcal{P}^u}^{inter}(q,q')\\
    \times \left(\phi_{\mathcal{P}^u}^{intra}(q,q')-\phi_{\mathcal{P}^u}^{intra}(q)-\phi_{\mathcal{P}^u}^{intra}(q')\right)
\end{multline}}
where $\phi_{\mathcal{P}^u}^{inter}(q,q')$ and $\phi_{\mathcal{P}^u}^{intra}(q,q')$ are defined similarly than $\phi_{\mathcal{P}^u}^{inter}(q)$ in \eqref{eq:phi_P_inter} and $\phi_{\mathcal{P}^u}^{intra}(q)$ in \eqref{eq:phi_P_intra} by replacing $G_{\textrm{GPP}}^{[u]}(v(q,u,R))$ by $G_{\textrm{GPP}}^{[u]}(v(q,u,R)) + G_{\textrm{GPP}}^{[u]}(v(q',u,R))$.
\end{theorem}
\begin{proof}
    The proof is provided in Appendix~\ref{sec:M2_proof}.
\end{proof}

\subsection{Joint UL EMFE/Coverage Performance}
The third part of this study involves an analysis of UL performance, considering both EMFE and coverage.  Not all UEs used to calculate the UL EMFE must be taken into account for the study of the coverage because only a subset uses the same subframes as explained in Subsection~\ref{ssec:topology}. Unlike the DL, EMFE and SINR are calculated at different points: The study of UL coverage is conducted from the perspective of the serving BS, while UL EMFE is studied from the point of view of the typical user. It is therefore adequate to make the approximation that UL EMFE and UL SINR are independent, with a small correction consisting in replacing the PP of interfering UEs by an I-PPP with density function~\eqref{eq:I_density_u} which accounts for their PV distribution. The joint CDF of the UEC is then given by the product of the EMFE CDF and SINR CCDF, as shown in Theorem~\ref{th:jointUL}. The error made will be analyzed in Subsection~\ref{ssec:jointUL}.

\begin{theorem}\label{th:jointUL} The joint CDF of UEC in a network where BSs are distributed according to a H-PPP and UEs are distributed according to a MCP+GPP mixture for the propagation model in \eqref{eq:prop_UL} and in \eqref{eq:prop_DL} is
\begin{align}
\begin{split}
        G^{u}(T_{c}^{u}, T_{e}) &= \mathds E_{R_0}\left[\mathbb{P}\left[\text{\normalfont{SINR}}^{u}_0 > T_{c}^{u}, \mathcal{P}^{u} < T_{e}\right]|R_0\right]\\
        &= F_{\textrm{emfe}}^{u}(T_{e}) \times F_{\textrm{cov}}^{u}(T_{c}^{u})
\end{split}
\end{align}
    where
{\smalltonormalsize
\begin{equation}\label{eq:cdf_cov_u}
        F^{u}_{\textrm{cov}}(T_{c}^{u}) \triangleq \mathds E_{R_0}\left[F^{u}_{\textrm{cov}}(T_{c}^{u}|R_0)\right] =\! \int_{r_e}^{\tau} \!\!F^{u}_{\textrm{cov}}(T_{c}^{u}|r_0) f_{R_0}^\beta(r_0) dr_0,
\end{equation}\vspace{-1em}
\begin{align}
\begin{split}
    &F^{u}_{\textrm{cov}}(T_{c}^{u}|R_0) \triangleq \mathbb{P}\left[\textrm{\normalfont SINR}_0^{u} > T_{c}^{u}\right]\\
    &\quad= \mathcal{L}_{N_{\sigma}^{u}}\left(\frac{T_{c}^{u} }{\Bar{S}^{u}(R_0, R_0)}\right) \!\mathcal{L}_{I_0^{u}}\left(\frac{T_{c}^{u}}{\Bar{S}^{u}(R_0,R_0)}\Big|R_0\!\right),
\end{split}
\end{align}\vspace{-1em}
\begin{multline}
    \mathcal{L}_{I_0^{u}}(s|R_0) \\
    = 
    \exp\left(\int_{R_0}^\tau \!\!\left(\int_{r_e}^{r}\!\!\frac{f_{Ry}(v|r)}{1+s \Bar{S}^{u}(v,r)} dv\!-\!1 \right)2\pi \lambda^{IU}(r)r dr\right)
\end{multline}}
and $\mathcal{L}_{N_{\sigma}^{u}}(s) = e^{-s\,N_{\sigma}^{u}}$.
\end{theorem}
\begin{proof}
    The independence between $\text{\normalfont{SINR}}^{u}_0$ and $\mathcal{P}^{u}$ allows one to express the joint CDF as the product between the CDF of UL EMFE and the CCDF of the UL SINR. This latter term, denoted as $F_{\textrm{cov}}^{u}$ in the equations, is obtained in a similar fashion to what was used for the meta distribution of UL SIR in \cite{uplink_meta}, the CCDF of SIR in \cite{Singh15} or the CCDF of SINR\cite{Bai16}. The differences between the current model and the cited works is the maximal power constraint in the FPC strategy and a different inhomogeneous density function \eqref{eq:I_density_u}.
\end{proof}
Note that because of the maximal power constraint, $\int_{r_e}^{r}\!\!\frac{f_{Ry}(v|r)}{1+s \Bar{S}^{u}(v,r)} dv$ can be written $\int_{r_e}^{r}\!\!\frac{f_{Ry}(v|r)}{1+s \Bar{S}^{u}(v,r)} dv\, \mathds 1\left[r \leq r_{m}\right]+ \left[\int_{r_e}^{r_{m}}\!\!\frac{f_{Ry}(v|r)}{1+s \Bar{S}^{u}(v,r)} dv + \frac{F_{R_0}^\beta(r_{m}, r)/F_{R_0}^\beta(r)}{1+s \Bar{S}^{u}(r_{m},r)}\right]\mathds 1\left[r \geq r_{m}\right]$.

\subsection{Joint EM Pollution/UL-DL Coverage Performance}
The forth part of the study consists of an analysis of the network performance, both for the UL and the DL. Compared to the previous subsection, DL EMFE and DL SINR are dependent because they are calculated at the same point and DL interference is equal to the DL intercell EMFE. The joint CDF of EMP-UDC is then given in Theorem~\ref{th:joint_all}.
\begin{theorem}\label{th:joint_all}The joint  CDF of EMP-UDC in a network where BSs are distributed according to a H-PPP and UEs are distributed according to a MCP+GPP mixture for the propagation model in \eqref{eq:prop_UL} and in \eqref{eq:prop_DL} is given by
{\small
\begin{align}\label{eq:EMPUDC}
\begin{split}
    &G(T_{c}^{u},T_{e},T_{c}^{d})\\
    & \triangleq \mathds E_{R_0}\left[\mathbb{P}\left[\text{\normalfont{SINR}}^{u}_0 > T_{c}^{u}, \mathcal{P}^{tot} < T_{e}, \text{\normalfont{SINR}}^{d}_0 > T_{c}^{d}\right]\Bigg|R_0\right]\\
    &= \int_{r_e}^\tau \!\! F^{u}_{\textrm{cov}}(T_{c}^{u}|r_0)  M(T_{e},T_{c}^{d}|r_0)  f_{R_0}^\beta(r_0) dr_0
\end{split}
\end{align}}
where
\begin{multline}
    M(T_{e},T_{c}^{d}|R_0) = 
    \left[ \frac{1}{2} - \frac{1}{2}  \exp\left(\frac{T_{e}}{P^{d}\,l^{d}(R_0)}\right)\right.\\
    \left. - \int_{0}^{\infty} \frac{1}{\pi q}\textrm{\normalfont Im}\left[ \phi_{I_0^{d}}(q|R_0)\,\zeta(q, T_{c}^{d},T_{e},l^{d}(R_0))\right]dq\right],
\end{multline}
\vspace{-1em}
{\footnotesize
\begin{multline}
    \zeta(q, T_{c}^{d}, T_4, \Bar{S}^{d}(R_0)) \!=\! \frac{1\!-\!\exp\left(\frac{-T'}{\Bar{S}^{d}(R_0)} \!\left(1\!+\!j\frac{q\Bar{S}^{d}(R_0)}{T_{c}^{d}}\right)\right)}{1+j\frac{q \Bar{S}^{d}(R_0)}{T_{c}^{d}}} e^{jqN_{\sigma}^{d}}\\
    + \phi_{\mathcal{P}^{u}}(q)e^{-jqT_{e}} \frac{\exp\left(\frac{T_{e}\left(jq \Bar{S}^{d}(R_0)-1\right)}{\Bar{S}^{d}(R_0)} \right) \!-\! \exp\left(\frac{T' \left(jq\Bar{S}^{d}(R_0)-1\right)}{\Bar{S}^{d}(R_0)}\right)}{jq\Bar{S}^{d}(R_0)-1}\!.
\end{multline}} 
\end{theorem}
\begin{proof}
    Because of the assumptions made before, $\mathbb{P}\left[\text{\normalfont{SINR}}^{u}_0 > T_{c}^{u}, \mathcal{P}^{tot} < T_{e}, \text{\normalfont{SINR}}^{d}_0 > T_{c}^{d}\right]$ can be written $\mathbb{P}\left[\text{\normalfont{SINR}}^{u}_0 > T_{c}^{u}\right]\!\times \!\mathbb{P}\left[\mathcal{P}^{tot} < T_{e}, \text{\normalfont{SINR}}^{d}_0 > T_{c}^{d}\right]$. The first term of the product is \eqref{eq:cdf_cov_u}. The second term is very similar to $\mathbb{P}\left[\mathcal{P}^{d} < T_{e}, \text{\normalfont{SINR}}^{d}_0 > T_{c}^{d}\right]$ calculated in \cite{GontierTWC}. The only difference is that $\mathcal{P}^{d}$ is replaced by $\mathcal{P}^{tot}$. The problem can be solved by writing $\mathbb{P}\left[ {S_0^{d}}+{I_0^{d}} < T'_{e}, \text{\normalfont{SINR}}^{d}_0 > T_{c}^{d}\right]$ with $T'_{e}\! =\! T_{e}-\mathcal{P}^{u}$.
\end{proof}

For the ease of analysis, we propose in Lemma~\ref{lem:cond_joint} the conditional joint CDF of EMP-UDC, conditioned on specific values of UL SINR thresholds and EMFE limits. This is particularly useful for legislators to observe the impact of a more stringent EMFE limit on the network performance. 
\begin{lemma}\label{lem:cond_joint}
    The joint CDF of EMP-UDC, conditioned on specific values of UL SINR thresholds and EMFE limits is
    \begin{equation}
        H(T_{c}^{d}|T_{c}^{u}, T_{e}) = \frac{G(T_{c}^{u},T_{e},T_{c}^{d})}{F_{\textrm{cov}}^{u}\left(T_{c}^{u}\right)\times F_{\textrm{emfe}}^{tot}\left(T_{e}\right)}
    \end{equation}
    where
    \begin{multline}
    F^{d}_{\textrm{cov}}(T_{c}^{d}) \triangleq \mathds E\left[\mathbb{P}\left[\textrm{\normalfont SINR}_0^{DL} > T_{c}^{d}\right]\right]\\
    = \int_{r_e}^{\tau} e^{-T_{c}^{d} N_{\sigma}^{d} /\Bar{S}^{d}(r_0)} \phi_{I_0^{d}}\left(j T_{c}^{d}/\Bar{S}^{d}(r_0)\Big|r_0\right) \, f_{R_0}^\beta(r_0)\,dr_0.
\end{multline}
\end{lemma}
\begin{proof}
    This metric is obtained from Bayes'rule, and by taking advantage of the hypothesis of the independence between UL SINR and UL/DL EMFE.
\end{proof}
It is worth noting that a joint CDF conditioned on specific values of UL and DL SINR could also be defined but this would require the additional calculation of a joint metric of UL and DL SINR.

\section{Numerical Results}
\label{sec:results}
In this section, the four questions of the introduction are answered, using the metrics defined in Section~\ref{sec:analytical_results}. Before going deeper, the term \textit{densification} has to be clarified. The number of BSs, clusters or UEs in the clusters can independently increase. To evaluate the impact of each, four scenarios are used. 
\begin{enumerate}[(a)]
    \item \textbf{Densification of clusters with constant BS density and number of points per cluster}
    \item \textbf{Densification of points in each clusters with constant cluster and BS densities}
    \item \textbf{Densification of BSs and clusters with a constant UE/BS density ratio $\lambda^{u}/\lambda^{b}$}
    \item \textbf{Densification of BSs and points in each clusters with a constant UE/BS density ratio $\lambda^{u}/\lambda^{b}$}
\end{enumerate}

In real-world scenarios, as BSs are typically macro cellular BSs (MCs) for low densities of BSs/UEs, they are more likely to be small cells (SCs) for high BSs/UEs densities for two reasons. First, as the density of BSs increases, it is of interest to decrease the emit power at the BSs because the emit power and the BS density are proportional to the mean DL EMFE, as can be seen from \eqref{eq:mean_exp_DL}. Second, there is a maximum number of UEs that a single BS can accommodate. We set the maximal number of UEs per BS as the number of RBs available over the channel bandwidth given by
\begin{equation}
    N_{RB} = \frac{B_w-2 \,\text{GB}}{\text{SCS} \cdot \text{subframes per frame}}.
\end{equation}
At 2.6~GHz, with a channel bandwidth of 20~MHz, a guardband (GB) of 5\% of the channel bandwidth at each extremity, a subcarrier spacing (SCS) of 15~kHz and 12 subframes per frame, the maximal number of RBs is $N_{RB} = 100$. Therefore, when $\lambda^{u} > 100 \lambda^b$ or $\lambda^b > \numprint[BS/km^2]{1000}$, we switch from a network of MCs to a network of SCs.
\begin{table}[h!]
    \begin{center}
    \begin{tabular}{|c|c||c|c|} 
     \hline
     $f_u$ & \numprint[GHz]{2.560} & $f_d$ & \numprint[GHz]{2.680}\\ 
     $B_w^u$ & \numprint[MHz]{20} & $B_w^d$ & \numprint[MHz]{20}\\ 
     $r_e$ & \numprint[m]{1} & $\tau$ & \numprint[km]{30}\\
     $r^{[p]}$ & \numprint[m]{100} & $p^{[o]}$ & 0.5\\
     $P_{u, max}$ & 23~dBm & $\upepsilon$ & 0.4\\
     $N_{\sigma}^{d} = N_{\sigma}^{u}$ & \numprint[dBm]{-95.40} & SNR$_{cell}^{u}$ & 3\\
     $f^{[o]}$ & $\frac{\exp\left({\frac{-(x-\numprint[m]{1})^2}{2 \cdot (\numprint[m]{1})^2}}\right)}{2\pi \cdot (\numprint[m]{1})^2}$ &  & \\
     \hline
    \end{tabular}
\vspace*{1 em}

    \begin{tabular}{ |c|c|c| } 
     \hline
     Parameter & MC & SC\\
     \hline
     $\alpha$ & 3.25 & 2.5\\
     $z$ & 33~m & 4m\\
     $P_{d}$ & 66~dBm & 40~dBm\\
     \hline
    \end{tabular}
\vspace*{1 em}

    \begin{tabular}{|c|c|c|c|} 
     \hline
     Scenario & $\lambda^b$ & $\lambda^{[gp]}$ & $\lambda^{[p]}$\\
     \hline
     \multirow{2}{*}{(a)} & $\numprint[BS/km^2]{10}$ (MC) & \multirow{2}{*}{Varies} & \multirow{2}{*}{$\numprint[km^{-2}]{160}$}\\
     & $\numprint[BS/km^2]{1000}$ (SC) &  & \\
     \hline
     \multirow{2}{*}{(b)} & \numprint[BS/km^2]{10} (MC) & \multirow{2}{*}{\numprint[clusters/km^2]{2.5}} & \multirow{2}{*}{Varies}\\
     & \numprint[BS/km^2]{1000} (SC) &  & \\
     \hline
     (c) & Varies & $2.5\lambda^b$ & \numprint[km^{-2}]{160}\\
     \hline
     (d) & Varies & \numprint[clusters/km^2]{25} & $\lambda^b\cdot 16$\\
     \hline
    \end{tabular}
    \end{center}
    \caption{Simulation parameters}
    \label{tab:sim_param}
\end{table}
The network parameters are listed in Table~\ref{tab:sim_param} with, first, the parameters that are unchanged for all simulations, second, the parameters changing for MCs and SCs, chosen according \cite{GontierAccess} for MCs and \cite{GSMAsmallcell} for SCs and third, the different densities for each scenario. The UL open-loop power $P_0^u = \textrm{SNR}_{cell}^{u} N_{\sigma}^{u} \kappa_u^{1-\upepsilon} \left(\frac{1}{16\lambda^u}+z^2\right)^{(1-\upepsilon)\alpha/2}$ is adapted so that the average UL SNR is always 3 for cell-edge users. This quantity remains constant for constant BS density. The bandwidth of 20~MHz is a typical value for the 2.6~GHz frequency band. The value of the FPC coefficient is chosen after typical values obtained in \cite{Ubeda08, Simonsson08}. The maximal UL transmit power is another typical value confirmed in \cite{Joshi17}. The value $\tau$ is chosen in order to satisfy the conditions $\tau >> R$ and $r_m < \tau$.

\subsection{Analysis of EM Pollution}
\label{ssec:UL_EMF_exp}

To address the primary objective of this study, we initiated a comparison between UL and DL EMFEs. This assessment aims to provide insights into whether it would be advisable for a legislator to consider incorporating UL considerations into future regulations. 
\begin{figure*}
\centering
\begin{minipage}{.33\textwidth}
  \centering
  \includegraphics[width=0.9\linewidth, trim={5.5cm, 9cm, 5.9cm, 9cm}, clip]{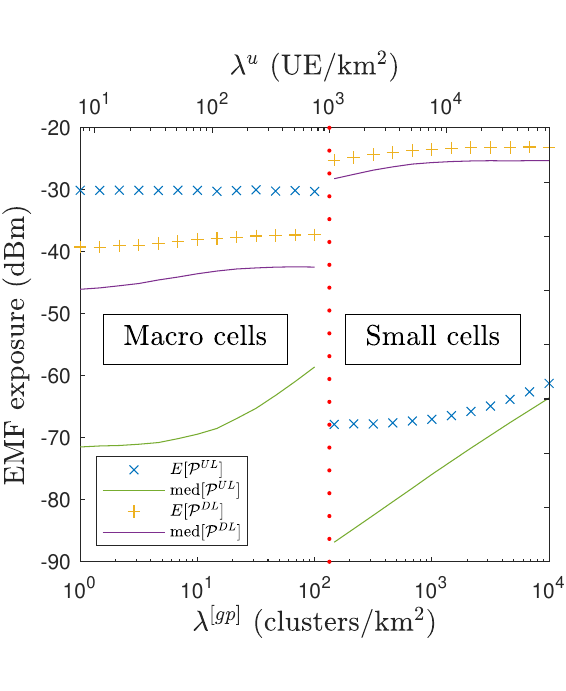}
  \captionof{figure}{Mean and median UL/DL EMFE in densification scenario~(a).}
\label{fig:ExpUL_GP_density_stat}
\end{minipage}%
\begin{minipage}{.33\textwidth}
  \centering
\includegraphics[width=0.9\linewidth, trim={5.5cm, 9cm, 5.9cm, 9cm}, clip]{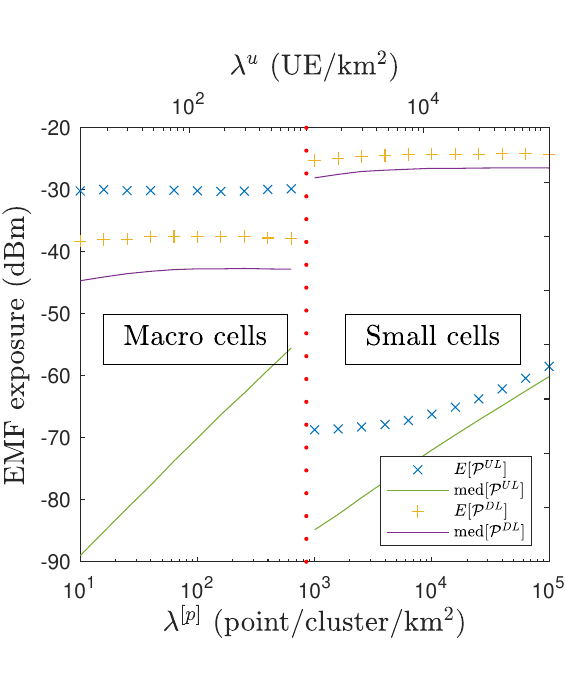}
  \captionof{figure}{Mean and median UL/DL EMFE in densification scenario~(b).}
  \label{fig:ExpUL_PA_density_stat}
\end{minipage}%
\begin{minipage}{.33\textwidth}
  \centering
  \includegraphics[width=0.9\linewidth, trim={5.5cm, 9cm, 5.9cm, 9cm}, clip]{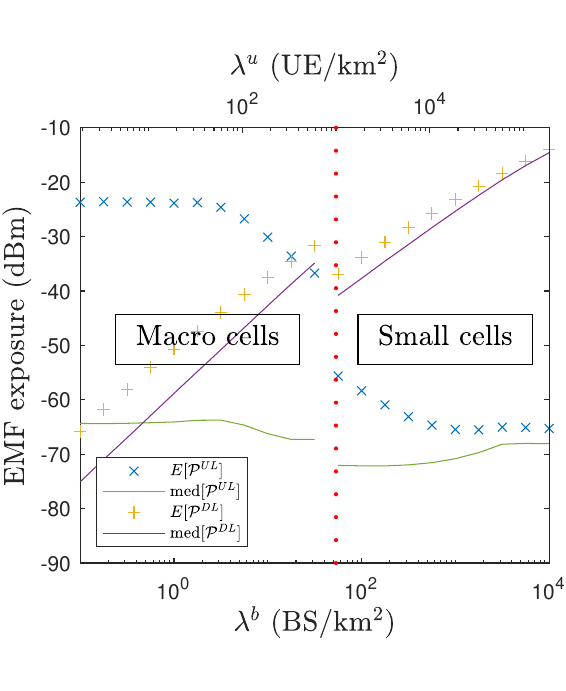}
  \captionof{figure}{Mean and median UL/DL EMFE in densification scenario~(c).}
\label{fig:ExpUL_BS_GP_density_stat}
\end{minipage}
\end{figure*}
Figs.~\ref{fig:ExpUL_GP_density_stat}, \ref{fig:ExpUL_PA_density_stat}, \ref{fig:ExpUL_BS_GP_density_stat} and \ref{fig:ExpUL_BS_PA_density_stat} illustrate the mean and median of UL and DL EMFEs for scenario~(a) to (d). The mean is obtained using the mathematical expressions in Corollary~\ref{cor:mean_tot} while the median is obtained via Monte-Carlo simulations (MCSs). In scenario~(a), for MCs, the mean UL EMFE is notably high at -30~dBm, which is 7 to 9~dBm higher than the mean DL EMFE and 30 to 40~dBm higher than the median UL EMFE. This outcome is not surprising due to the clustered nature of the network. In instances where another UE is near the typical user, EMFE values spike. Since BSs are generally distant from UEs, many UEs operate near their maximum power, which is not observed in a network of small cells SCs despite higher UE density. 

The median UL EMFE shows that most of the time, DL EMFE is significantly higher than UL EMFE. The analysis also reveals an increase in both mean and median DL EMFE with UE densification. This slight increase is because some cells may have no UEs, causing the corresponding BS to remain inactive. A smaller $\lambda^{u}/\lambda^{b}$ ratio results in a higher $\nu = (\gamma/(\gamma+\delta))^\gamma$ and a smaller $\lambda^{r}$ compared to $\lambda^{b}$. Scenario~(b) shows similar results to scenario~(a). The main difference is the median UL EMFE for MCs. Despite a high number of clusters in scenario~(b), low UE density within each cluster keeps UL EMFE low. This demonstrates that significant UL EMFE only occurs within the same cell. As cluster density or UE density per cluster increases, scenarios~(a) and (b) converge because high cluster density leads to merging clusters, increasing UEs in the typical user's cluster.

\begin{figure*}
\centering
\begin{minipage}{.33\textwidth}
  \centering
\includegraphics[width=0.9\linewidth, trim={5cm, 9cm, 5.9cm, 9cm}, clip]{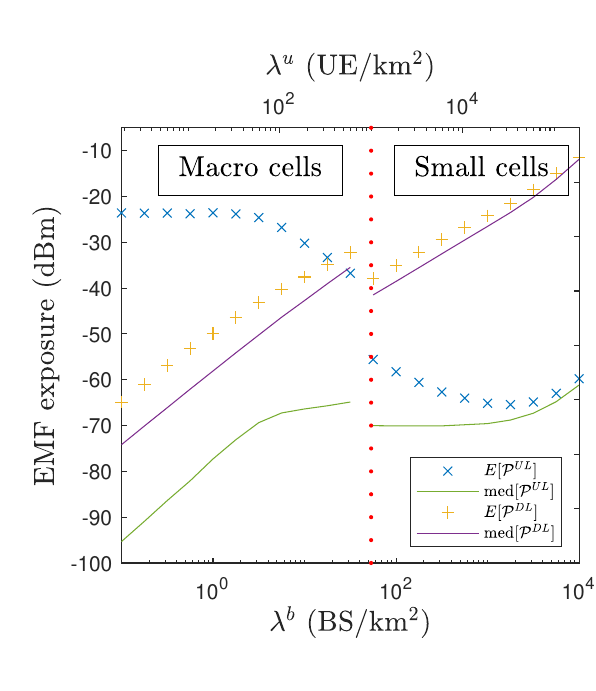}
  \captionof{figure}{Mean and median UL/DL EMFE in densification scenario~(d).}
  \label{fig:ExpUL_BS_PA_density_stat}
\end{minipage}%
\begin{minipage}{.33\textwidth}
  \centering
\includegraphics[width=0.9\linewidth, trim={5cm, 9cm, 5.9cm, 9cm}, clip]{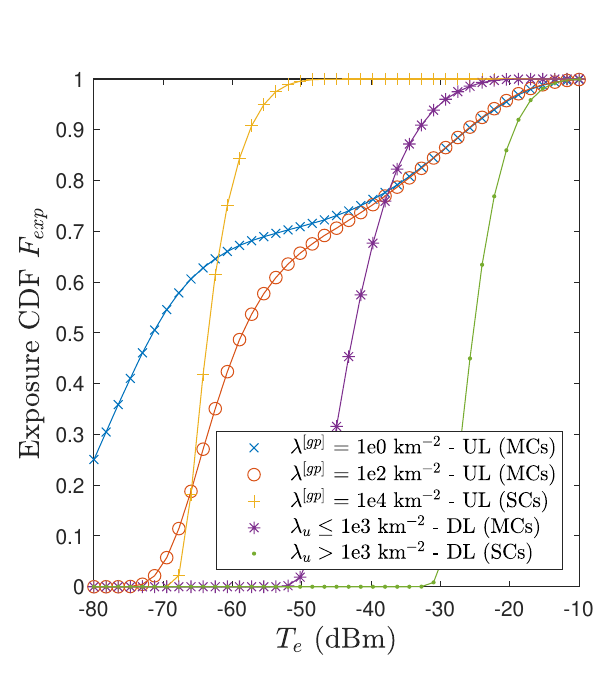}
  \captionof{figure}{Comparison of CDFs of UL and DL EMFEs, scenario~(a)}
  \label{fig:Exp_GP_density_cdf}
\end{minipage}%
\begin{minipage}{.33\textwidth}
  \centering
\includegraphics[width=0.9\linewidth, trim={5cm, 9cm, 5.9cm, 9cm}, clip]{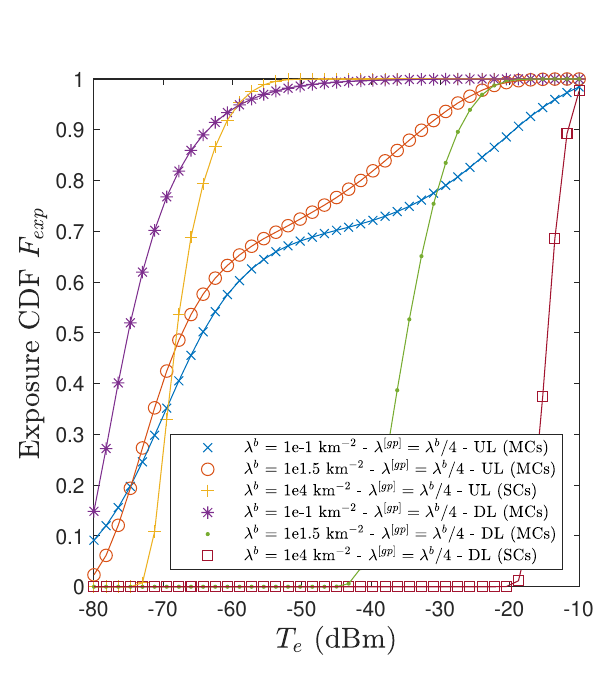}
  \captionof{figure}{Comparison of CDFs of UL and DL EMFEs, scenario~(c)}
  \label{fig:Exp_BS_GP_density_cdf}
\end{minipage}
\end{figure*}
In scenarios~(c) and (d), the mean and median DL EMFE are directly proportional to BS density. As BS density increases, the average distance between a UE and its serving BS decreases. Consequently, the FPC mechanism adjusts the UE's transmit power, explaining the inflection point in the mean and median UL EMFE curves. The difference in median UL EMFE between the scenarios is again due to significant UL EMFE occurring only within the same cell. 
In all scenarios, transitioning to a network of SCs as the UE/BS density increases proves beneficial for uplink EMFE, as it remains globally lower and avoids high-likelihood spiking values. However, in scenarios~(a) and (c), DL EMFE increases with SCs despite the lower BS transmit power. This can be attributed to the proportional relationship between mean DL EMFE and both transmit power and BS density. Specifically, while the transmit power decreases by a factor of $10^{2.6}$, the BS density increases by a factor of $10^3$. Conversely, in scenarios~(b) and (d), DL EMFE decreases with SCs due to the chosen network parameters.

Given similar conclusions from scenarios~(a) and (b) and scenarios~(c) and (d), the study will focus on scenarios~(a) and (c). Consistent findings are evident in the CDFs of DL and UL EMFE in Figs. \ref{fig:Exp_GP_density_cdf} and \ref{fig:Exp_BS_GP_density_cdf} for scenario~(a) and (c), respectively, derived from Lemma~\ref{lem:DL_EMF} and Theorem~\ref{eq:expucdf}. 
MCSs and mathematical expressions match in both figures. 

In scenario~(a) for MCs, for the chosen parameters, there is a 77\% probability that both UL and DL EMFE are below $T_e = \numprint[dBm]{-38}$, regardless of UE density. The total EMFE, not shown for clarity, combines both CDFs, being influenced mainly by DL EMFE for $T_e < \numprint[dBm]{-38}$ and by UL EMFE for $T_e > \numprint[dBm]{-38}$. For SCs, the total EMFE almost equals DL EMFE. Comparing scenarios~(a) and (c), UL EMFE increases with the number of clusters but decreases as the number of BSs proportionally increases. The unusual CDF shapes for UL EMFE in MCs result from significant instances where UEs are very close to the typical user and transmit at high power, explaining why the mean UL EMFE can sometimes exceed the mean DL EMFE, though this does not occur for the medians.

\subsection{Meta Distribution of EMFE}
The second question of the introduction can be answered by examining Fig.~\ref{fig:UL_meta_Te}. This figure represents the meta distribution of UL, DL and total EMFE for two EMFE limits as defined in~\eqref{eq:meta_def}. The solid lines are generated using MCSs with $5000 \times 5000$ realizations. The markers represent numerical values obtained by using Proposition~\ref{prop:beta_approx}. The simulation parameters are taken from Table~\ref{tab:sim_param} with $\lambda^b = \numprint[BS/km^2]{10}$, $\lambda^{[gp]} = \numprint[clusters/km^2]{50}$ and $\lambda^{[p]} = \numprint[parents/km^2]{318}$, meaning that the average UE density is $\lambda^u = \numprint[UE/km^2]{750}$. A significant alignment occurs between MCSs and the analytical approximation for the UL and DL meta distribution. The discrepancy between the lines remains small for the total EMFE at $\mathcal{F}_{F_{\textrm{emfe}}^{tot}}(\numprint[dBm]{-27}, s)$ but becomes more pronounced as $T_e$ decreases as seen with $\mathcal{F}_{F_{\textrm{emfe}}^{tot}}(\numprint[dBm]{-37}, s)$. This indicates that the approximation is more accurate when there is a high probability that a large number of users are below a stringent threshold. Despite the complex mathematical expressions for the moments, using stochastic geometry (SG) proves advantageous due to the reduced computational time compared to MCSs. At $T_e \!= \!\numprint[dBm/m^2]{-37}$, over 90\% of UE locations have a 87\% probability of experiencing UL EMFE below this threshold. In other words, 90\% of users encounter EMFE below $T_e \!= \!\numprint[dBm/m^2]{-37}$ at least 90\% of the time. The same proportion of users is below it 41\% of the time for the DL, contrasting with only 34\% for total EM pollution. The different between DL and total EMFE arises again from high UL EMFE values occurring with a low probability. Similarly, for 90\% of users being below the threshold 90\% of the time, $T_e\! \geq \!\numprint[dBm/m^2]{-27}$ is required as shown by the upper right turquoise rectangle.
\begin{figure*}
\centering
\begin{minipage}{.33\textwidth}
  \centering
\includegraphics[width=0.9\linewidth, trim={5cm, 9cm, 5.9cm, 10cm}, clip]{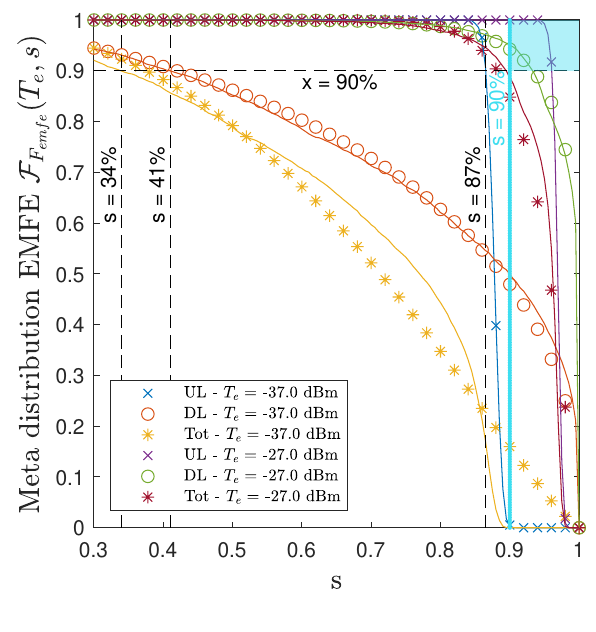}
  \captionof{figure}{Meta distribution of EMFE with UEs $\sim$ MCP + GPP (MCs)}
  \label{fig:UL_meta_Te}
\end{minipage}%
\begin{minipage}{.33\textwidth}
  \centering
\includegraphics[width=0.9\linewidth, trim={5cm, 9cm, 5.9cm, 10cm}, clip]{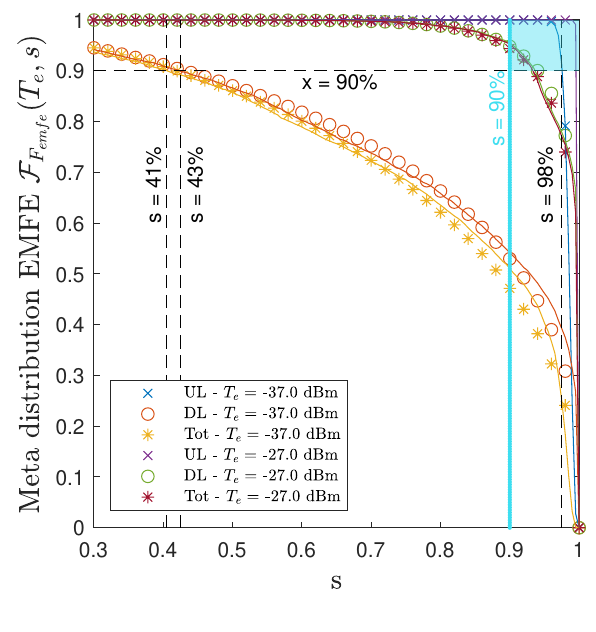}
  \captionof{figure}{Meta distribution of EMFE with UEs $\sim$ MCP (MCs)}
  \label{fig:UL_meta_MCP_Te}
\end{minipage}%
\begin{minipage}{.33\textwidth}
  \centering
\includegraphics[width=0.9\linewidth, trim={5cm, 9cm, 5.7cm, 10cm}, clip]{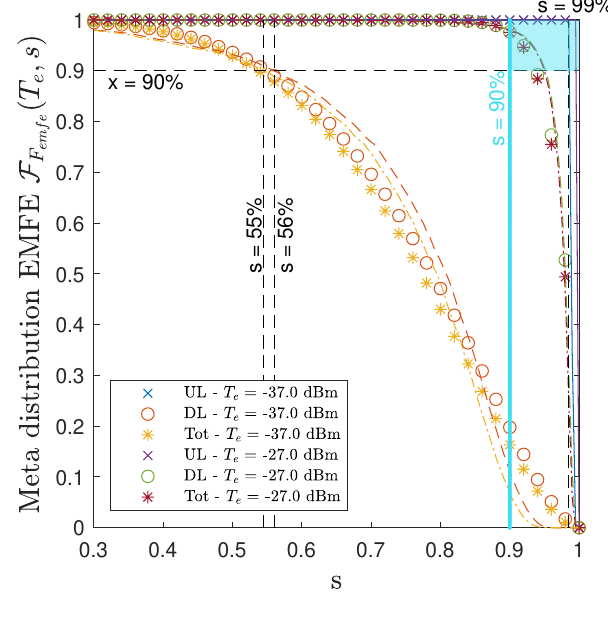}
  \captionof{figure}{Meta distribution of EMFE with UEs $\sim$ H-PPP (MCs)}
  \label{fig:UL_meta_HPPP_Te}
\end{minipage}
\end{figure*}
Fig.~\ref{fig:UL_meta_MCP_Te} depicts the meta distribution of EMFE for UEs distributed as a MCP, i.e. with $p^{[o]}$ = 0. To ensure a fair comparison, $\lambda^{[p]} = \numprint[parents/km^2]{477}$ maintaining the same total UE density as in the first case. In this figure, the two-moment approximation performs significantly better across all curves due to the simpler mathematical model. In  Fig.~\ref{fig:UL_meta_MCP_Te}, the two-moments approximation is much better for all curves because of the easier mathematical model. The distribution of DL and total EMFE are much closer than those in Fig.~\ref{fig:UL_meta_Te}, indicating that UL EMFE is several orders of magnitude lower than DL EMFE in almost all situations. This difference arises from the fact that neighboring users can position their second equipment closer to the typical user than their own position in the MCP+GPP model. When UEs follow a MCP distribution, even though they are clustered, their density within a cluster is uniform, implying no attraction or repulsion between UE locations. This contrasts with the MCP+GPP model, making these two models represent extreme cases of a network with clustered UEs. In the MCP+GPP network, 34\% of PU locations had a 34\% probability of experiencing EM pollution below $T_e != !\numprint[dBm/m^2]{-37}$. In contrast, in an MCP network, this probability increases to 41\%, and if UEs are distributed as a Homogeneous Poisson Point Process (H-PPP), the probability further rises to 55\%. This last value is shown in Fig.~\ref{fig:UL_meta_HPPP_Te}, which depicts the meta distribution for a network with UEs distributed as an H-PPP with $\lambda^u = \numprint[UE/km^2]{750}$.

\subsection{Joint UL EMFE/Coverage Performance}
\label{ssec:jointUL}
This section demonstrates the impact of BS and UE densification on network performance. It explores how to determine the optimal BS density to deploy, given a specific coverage threshold and an EMFE upper limit. It also highlights the type of BS to deploy depending on the density of UEs.

Fig.~\ref{fig:UL12_GP_density_joint} displays the joint CDF of UEC, derived from Theorem~\ref{th:jointUL}, for various values of $T_{e}$ (EMFE limit) and $T_{c}$ (coverage threshold) in densification scenario~(a). It is worth noting that the close proximity between MCSs (solid lines) and numerical values (markers) from the mathematical model validates the approximation of no correlation between UL EMFE and UL SINR.
In the case of a MC network, a bell-shaped curve is observed, indicating optimal performance at UE densities between approximately 80 to $\numprint[UE/km^2]{300}$ (10 to $\numprint[cluster/km^2]{40}$) when $\lambda^{b} = \numprint[BS/km^2]{10}$. As UE density increases, $T_{e}$ becomes increasingly significant. If the number of clusters or UEs rises too much, performance can rapidly decline depending on the EMFE limit. The right side of the graph shows that transitioning to a SC network is beneficial when UE density exceeds $\numprint[UE/km^2]{1000}$. Performance then remains stable until cluster density reaches several thousand per km². Specifically, if the EMFE limit is set to $\numprint[dBm]{-59}$, switching to a SC network is advantageous when UE density is below $\numprint[UE/km^2]{1000}$. This analysis allows for understanding network performance across different UE density ranges and determining when to switch the type of BS deployed.

Fig.~\ref{fig:UL12_BS_GP_density_joint} shows the same metric in densification scenario~(c), illustrating the optimal BS density to deploy based on coverage and EMFE limits. For example, with $T_{c} = \numprint[dB]{0}$ and $T_{e} = \numprint[dBm]{-45}$, depicted by the purple curve with star markers, the optimal BS density is $\numprint[BS/km^2]{3.2}$, resulting in a joint performance metric $G(\numprint[dB]{0},\numprint[dBm]{-45}) = 0.29$. A more stringent EMFE limit of $\numprint[dBm]{-59}$, shown by the yellow curve with '$+$' markers, reduces overall network performance to $G(\numprint[dB]{0},\numprint[dBm]{-59}) = 0.25$, with the optimal BS density associated with optimal performance shifting to the left, decreasing to [...]. At lower BS densities, the difference between the curves narrows when modifying the EMFE limit, indicating that coverage remains the limiting factor.
Additionally, starting again from $G(\numprint[dB]{0},\numprint[dBm]{-45})$, adopting a lower coverage threshold of \numprint[dB]{-5} results in an average 72\% improvement in network performance. Swapping from a MC network to a SC network improves performance when a stricter coverage threshold is applied. Regardless of the coverage threshold, having SCs instead of MCs is beneficial for network performance when density exceeds $\numprint[UE/km^2]{1000}$ ($\numprint[BS/km^2]{53}$), as it allows reaching a higher optimal BS density before performance declines. At $\numprint[dB]{0}$, optimal performance is achieved with a BS density of $\numprint[BS/km^2]{177}$ if $T_{e} = \numprint[dBm]{-45}$ or $\numprint[BS/km^2]{316}$ if $T_{e} = \numprint[dBm]{-59}$. At high SC densities, the EMFE limit has little influence on performance, while the coverage threshold remains a significant factor.

\begin{figure}
\centering
\includegraphics[width=0.65\linewidth, trim={5cm, 9cm, 5.9cm, 9.8cm}, clip]{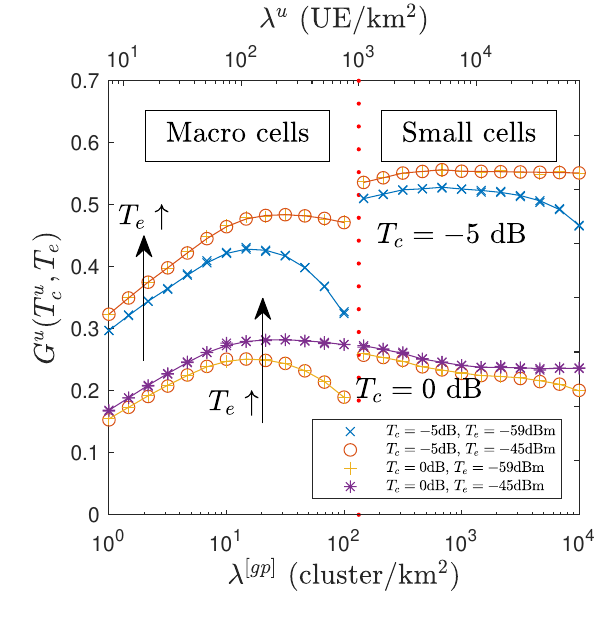}
\caption{Joint CDF of UEC for specific threshold values for scenario~(a)}
\label{fig:UL12_GP_density_joint}
\end{figure}

\begin{figure}
\centering
\includegraphics[width=0.65\linewidth, trim={5cm, 9cm, 5.9cm, 9.8cm}, clip]{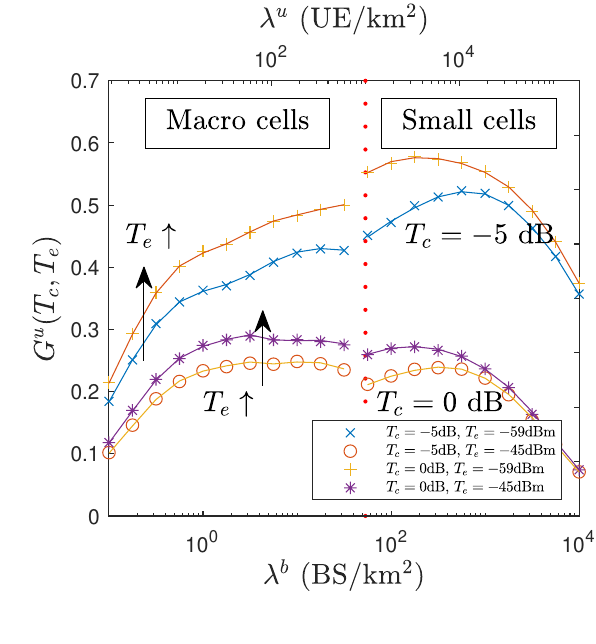}
\caption{Joint CDF of UEC for specific threshold values for scenario~(c)}
\label{fig:UL12_BS_GP_density_joint}
\end{figure}

\subsection{Joint EM Pollution/UL-DL Coverage Performance}
\label{ssec:jointULDL}
Finally, the conditional CDFs of EMP-UDC obtained from Lemma~\ref{lem:cond_joint} are displayed in Figs.~\ref{fig:UL12p54_GP_density_cond} and \ref{fig:UL12p54_BS_GP_density_cond} for scenarios~(a) and (c), respectively. These CDFs are conditioned on $T_{c}^{u} = \numprint[dB]{0}$ and $T_{e} = \numprint[dBm]{-31}$ for the sake of analysis and are presented for various values of DL SINR threshold. Here, both DL and UL EMFE are combined. The results demonstrate a strong correlation between the mathematical expressions and MCSs, with network performance decreasing as the BS or UE density increases. Transitioning to a network of SCs at high UE densities proves beneficial in scenario~(c) but not in scenario~(a). This discrepancy arises from the chosen network parameters, where DL EMFE increases with SCs in scenario~(a), as discussed in ~\ref{ssec:UL_EMF_exp}. These figures address the fourth question posed in the introduction. Users who meet the uplink SINR requirement of $\numprint[dB]{0}$ and experience an EM pollution below $\numprint[dBm]{0}$ will achieve a DL SINR threshold of $\numprint[dB]{-6}$ with over 80\% probability if the cluster density is below 0.8~cluster/km$^2$ (fewer than 50~UE/km$^2$) in scenario~(a) and if the BS density is below 5.6~BS/km$^2$ (fewer than 105~UE/km$^2$) in scenario~(c).

The low network performance is attributed to the strict EMFE limit chosen, $T_{e} = \numprint[dBm]{-31}$. When converted into IPD using~\eqref{eq:IPD_def}, this limit corresponds to $\numprint[W/m^2]{7.3e-4}$, which is significantly lower than the international ICNIRP recommendation of $\numprint[W/m^2]{10}$ \cite{icnirp}. A less strict EMFE limit would reduce the impact of EMFE constraints on network performance, leading to improved performance metrics. Additionally, this study considers only a single frequency band, whereas the ICNIRP recommendation applies to the cumulative EMFE from all frequency bands.

\begin{figure}
\centering
\includegraphics[width=0.65\linewidth, trim={5cm, 9cm, 5.9cm, 9.8cm}, clip]{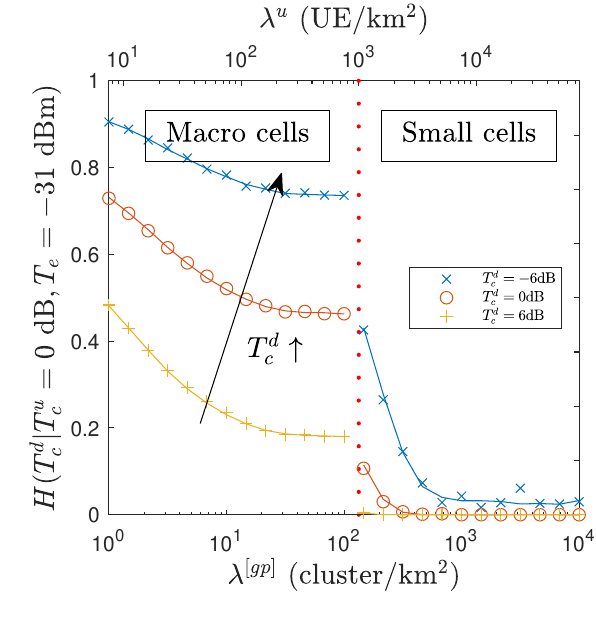}
\caption{Joint CDF of EMP-UDC conditioned on $T_{c}^{u} = \numprint[dB]{0}$ and $T_{e} = \numprint[dBm]{-31}$ in scenario~(a)}
\label{fig:UL12p54_GP_density_cond}
\end{figure}

\begin{figure}
\centering
\includegraphics[width=0.65\linewidth, trim={5cm, 9cm, 5.9cm, 9.8cm}, clip]{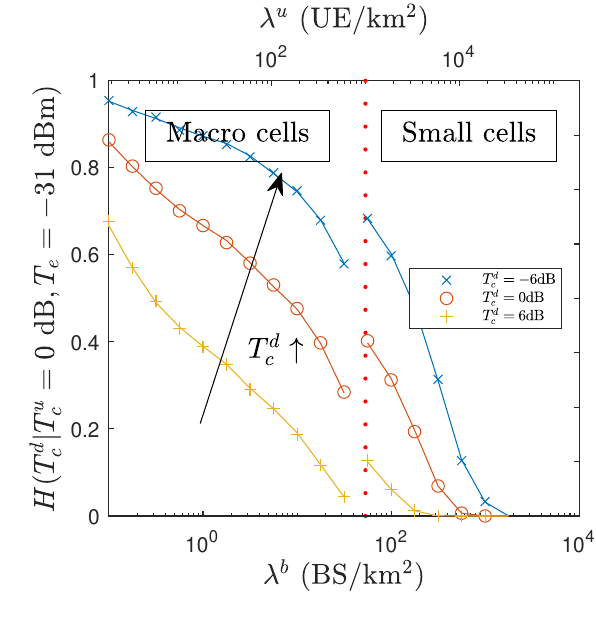}
\caption{Joint CDF of EMP-UDC conditioned on $T_{c}^{u} = \numprint[dB]{0}$ and $T_{e} = \numprint[dBm]{-31}$ in scenario~(c)}
\label{fig:UL12p54_BS_GP_density_cond}
\end{figure}

\section{Conclusion}
\label{sec:conclusion}
In this paper, we have developed mathematical expressions based on SG to evaluate network performance in both the UL and downlink DL in terms of EMFE and SINR. This analysis incorporates a truncated UL power control strategy in a network with clustered users, where each user may possess one or two UE devices. We introduced a general framework and provided mathematical expressions for various joint, marginal, and conditional metrics. Multiple scenarios were examined, including the densification of BSs, clusters of UEs, and UEs within each cluster.

Our findings indicate that UL EMFE generally remains several orders of magnitude lower than DL EMFE across most scenarios. Exceptions occur in instances with low UE and BS densities, where the high UE transmit power, coupled with the close proximity of other UEs, leads to elevated UL EMFE. The framework presented in this paper has demonstrated its robustness and versatility, as evidenced by the successful addressing of the four research questions posed in the introduction through both mathematical and numerical results.

It is important to highlight that this study represents a worst-case scenario where BSs transmit continuously at maximum power, UEs utilize all available resources, and EMFE limits are highly stringent. In practical deployments, more advanced strategies are employed. For instance, in 5G networks, adjacent BSs collaborate to allocate different subframes to their users at specific times, effectively reducing intercell interference and enhancing overall network performance. Future research could explore the impact of these allocation strategies.

Additionally, this study assumes all UEs are located outdoors. Integrating the performance of indoor UEs can be achieved by partitioning the UE point process into outdoor and indoor UEs with a specified probability. Indoor UEs would then be subject to a higher path loss exponent. This extension would provide a more comprehensive understanding of network performance in varied environments.

{\appendices

\section{Proof of Theorem~\ref{eq:expucdf}}\label{sec:expucdfproof}

According to its definition, the CF of UL EMFE is
\begin{equation}
    \phi_{\mathcal{P}^{u}}(q) = \mathds E\left[\exp\left(jq \sum_{i \in \Phi_{u} \setminus\left\{Y_0\right\}} P_i^{u}\, |\Tilde{h}_i|^2\, \Tilde{l}^{u}_i\right)\right].
\end{equation}
The proof is very similar to the proof of CFs of DL EMFE in \cite{GontierTWC} or UL EMFE in \cite{Qin24} and follows the steps of (i) converting the exponential of a sum to a product of exponentials (ii) taking advantage of the distibution of $|h|$ and of the pairwise independence of the $|h_i|$'s to convert each exponential into a function $v = (1-jq P_i^{u} \Tilde{l}^{u}_i)^{-1}$ and (iii) applying the probability generating functional (PGFL). In the case of a cluster process, intra- and intercell interference have to be distinguised. The PGFL then only applies on intercell interference. The PGFL of a function $v$ in a Poisson cluster process is
\begin{equation}
    G(v) = \exp\left(\lambda_c \int_{\mathbb{R}^2} \left(G_0^{[\vec{x}]}(v)-1\right)\,d\vec{x}\right)
\end{equation}
where $G_0^{[x]}(v) = \mathds E_{\Phi^{[x]}}\left[\prod_{i \in \Phi^{[x]}} v_i\right]$ is the PGFL of the cluster $\Phi^{[x]}$ that is centered at $x$. Applying this to the MCP+GPP mixture, this gives for the intercell interference
\begin{equation}
    \phi_{\mathcal{P}^u}^{inter}(q) = \exp\left(2 \pi\lambda^{[gp]}  \int_{r_c}^\tau \left(G_{\textrm{MCP}}^{[x]}(v)-1\right)x\,dx\right).
\end{equation}
The PGFL of a cluster of parents randomly distributed with a Poisson law at distance $X$ from the typical user and $R$ to the cluster's BS is $G_{\textrm{MCP}}^{[X]}(v) = \mathbb{E}_R[G_{\textrm{MCP}}^{[X]}(v|R)]$ with 
{\small
\begin{align}\label{eq:G_MCP_proof}
    \begin{split}
        &G_{\textrm{MCP}}^{[X]}(v|R)\\
        &\stackrel{(a)}{=} \sum_{n =0}^\infty e^{-\lambda^{[p]}\pi r_c^2}\left(\lambda^{[p]} \pi r_c^2\right)^n\frac{1}{n!} \\
        &\qquad \times\left(\int_0^{r_c+X}G_{\textrm{GPP}}^{[u]}(v|R)f_{R^{[p]}}(u|X)du\right)^n\\
        &\stackrel{(b)}{=} \!\exp\!\left(\!\lambda^{[p]}\pi r_c^2\left(\int_0^{r_c+X}\!G_{\textrm{GPP}}^{[u]}(v|R) f_{R^{[p]}}(u|X)du-1\right)\right)
    \end{split}
\end{align}}where $f_{R^{[p]}}(u|x)$ is used in $(a)$ to obtain the distance to a specific parent knowing the distance to the cluster center and $(b)$ is given by calculating the converging series. The expectation on $R$ can be applied by using the PDF \eqref{eq:fR} with $\beta \!=\! 1$ such that $G_{\textrm{MCP}}^{[X]}(v)\! =\! \int_{r_e}^{\tau}\!G_{\textrm{MCP}}^{[X]}(v|r)f_{R_0}^1(r) dr$.  The PGFL of a cluster in a GPP is given by $G_{\textrm{GPP}}^{[u]}(v|R)\! =\! (1-p^{[o]})v(u|R)+p^{[o]}v(u|R)\int_{-r_c}^{r_c}v(u+y|r)f^{[o]}(y)dy$. 

The expression of the PGFL of intracluster EMFE is very similar to~\eqref{eq:G_MCP_proof}, except that the summation starts at $n \!=\! 1$. With this change, the PGFL is $G_{\textrm{MCP}}^{[R^{[gp]}_0]}(v)\! =\! \int_{r_e}^{\tau}\!G_{\textrm{MCP}}^{[R^{[gp]}_0]}(v|r)f_{R_0}^1(r) dr$ with
{\small
\begin{multline}
        G_{\textrm{MCP}}^{[R^{[gp]}_0]}(v|R) 
        = \exp\!\left(-\lambda^{[p]}\pi r_c^2\right)\\
        \times\left(\exp\!\left(\!\lambda^{[p]}\pi r_c^2\int_0^{r_c+R^{[gp]}_0}\!G_{\textrm{GPP}}^{[u]}(v|R) f_{R^{[p]}}(u|R^{[gp]}_0)du\right)-1\right)\!.\!
\end{multline}}
The cluster is centered at $R^{[gp]}_0$, whose PDF is \eqref{eq:pdf_d_cluster} with $r_c > R^{[gp]}_0$. This finishes the proof of Theorem~\ref{eq:expucdf}.

\section{Proof of Theorem~\ref{th:M2}}
\label{sec:M2_proof}
     The first step consists of expanding $\mathcal{M}_{2}(T_e)$ by using Lemma~\ref{lem:CDF_cond} and distributing the expectation operator. This gives
 {\small
\begin{multline}\label{eq:M2_first}
    \mathcal{M}_{2}(T_{e}) = \frac{1}{4} - \mathbb E_{\Psi_b, \Phi^{[gp]}}\left[\int_{0}^{\infty} \textrm{\normalfont Im}\left[\phi_{\mathcal{P}^{tot}}(q|\Psi)\,e^{-jqT_{e}}\right]\frac{1}{\pi q}\,dq\right]\\
    + \mathbb E_{\Psi_b, \Phi^{[gp]}}\left[\frac{1}{\pi^2}\left(\int_{0}^{\infty} \textrm{\normalfont Im}\left[\phi_{\mathcal{P}^{tot}}(q|\Psi)\,e^{-jqT_{e}}\right]\,q^{-1}\,dq\right)\right.\\
    \left.\times\left(\int_{0}^{\infty} \textrm{\normalfont Im}\left[\phi_{\mathcal{P}^{tot}}(q'|\Psi)\,e^{-jq'T_{e}}\right]\,q'^{-1}\,dq'\right)\right].
\end{multline}}The second term in \eqref{eq:M2_first} corresponds to $\mathcal{M}_{1}(T_{e})-1/2$. By letting $\Omega(T_e)$ be the third term in \eqref{eq:M2_first}, \eqref{eq:M2} is obtained. 

The second step consists of interchanging the integrals and the expectation operator in the expression of $\Omega(T_e)$, so that it can be written as in \eqref{eq:Omega} with $
    \omega(T_e, q, q') \!= \!\mathbb E_{\Psi_b, \Phi^{[gp]}}[\textrm{\normalfont Im}[\phi_{\mathcal{P}^{tot}}(q|\Psi)e^{-jqT_{e}}]\textrm{\normalfont Im}[\phi_{\mathcal{P}^{tot}}(q'|\Psi) e^{-jq'T_{e}}]]$.

The third step focuses on the expansion of $\omega$. By using $\textrm{\normalfont Im}[x]\! = \!(x-\Bar{x})/{2}$ and $\textrm{\normalfont Re}[x]\! = \!({x+\Bar{x}})/{2}$ and by writing $b_1 \!=\! \exp(-jqT_{e})$ and $b_2\! = \!\exp(-jq'T_{e})$, it is obtained that $
    \omega(T_e, q, q')
    = \frac{1}{2}\,\textrm{\normalfont Re}\left[\gamma_+(q, q')\,b_1\,b_2\right]-\frac{1}{2}\,\textrm{\normalfont Re}\left[\gamma_-(q, q')\,b_1\,\Bar{b_2}\right]
$
with $\gamma_+(q, q')\! = \!\mathbb E_{\Psi_b, \Phi^{[gp]}}\left[\phi_{\mathcal{P}^{tot}}\left(q | \Psi\right)\,\phi_{\mathcal{P}^{tot}}\left(q' | \Psi\right)\right]$ and $\gamma_-(q, q') = \mathbb E_{\Psi_b, \Phi^{[gp]}}\left[\phi_{\mathcal{P}^{tot}}\left(q | \Psi\right)\,\Bar{\phi}_{\mathcal{P}^{tot}}\left(q' | \Psi\right)\right]$. 
The last step is the resolution of $\gamma_\pm(q, q')$.
By replacing the CFs by their expression given in \eqref{eq:cf_pd_cond} and \eqref{eq:cf_pu_cond}, by distributing the product and by using the PGFL, we obtain for $\gamma_+(q, q') = \gamma_+^d(q, q')\gamma_+^u(q, q')$ with
\begin{multline}
    \gamma_+^d(q, q') \\
    = \exp\!\left(2\pi \lambda^b\int_{r_e}^{\tau}\left(\frac{1}{1- j q P^d l(x)^d}\frac{1}{1- j q' P^d l(x)^d}-1\right)dx\right)
\end{multline}and
{\footnotesize
\begin{align}
    \begin{split}
        &\gamma_+^u(q, q') = \mathbb{E}_{\Psi_b, \Phi^{[gp]}}\left[\exp\!\left(-2\lambda^{[p]}\pi r_c^2\right)\right.\\
        &\left.\left(\exp\!\left(\!\lambda^{[p]}\pi r_c^2\int_0^{r_c+R^{[gp]}_0}\!G_{\textrm{GPP}}^{[u]}(v|R) f_{R^{[p]}}(u|R^{[gp]}_0)du\right)-1\right)\!\right.\\
        &\left.
        \times\left(\exp\!\left(\!\lambda^{[p]}\pi r_c^2\int_0^{r_c+R^{[gp]}_0}\!G_{\textrm{GPP}}^{[u]}(v'|R) f_{R^{[p]}}(u|R^{[gp]}_0)du\right)-1\right)\!\right.\\
        &\times\!\left.\!\prod_{i \in \Phi^{[gp]}\setminus{R^{[gp]}_0}}\!\exp\!\left(\!\lambda^{[p]}\pi r_c^2\int_0^{r_c+X_i}\!G_{\textrm{GPP}}^{[u]}(v|R) f_{R^{[p]}}(u|X_i)du\right)\right.\\
        &\left.\!\times \!\exp\!\left(-2\lambda^{[p]}\pi r_c^2\right)\exp\!\left(\!\lambda^{[p]}\pi r_c^2\int_0^{r_c+X_i}\!G_{\textrm{GPP}}^{[u]}(v'|R) f_{R^{[p]}}(u|X_i)du\right)\right]
    \end{split}
\end{align}}with $v'(q,u|R) = v(q',u|R)$. By applying the expectation on $R$ on each cluster using the PDF \eqref{eq:fR} with $\beta \!=\! 1$ on each cluster, then by applying the PGFL for the intercell part and PDF \eqref{eq:pdf_d_cluster} for the intracell part, and by rearraging the terms, \eqref{eq:gamma+u} is obtained. The same reasoning can be applied to $\gamma_{-}(q,q')$.

}


\bibliographystyle{IEEEtran}
\bibliography{bibli}

\vfill
\end{document}